\documentclass[a4paper,UKenglish,cleveref,autoref,thm-restate]{lipics-v2021}

\pdfoutput=1 
\hideLIPIcs  

\graphicspath{{figures/}}
\usepackage{amssymb, amsfonts}
\usepackage{cite}
\usepackage{bold-extra} 
\crefname{algocf}{Algorithm}{Algorithm}
\usepackage{verbatim}
\usepackage[ruled]{algorithm2e}
\usepackage{mathtools}
\usepackage{tikz}
\SetKwComment{Comment}{$\triangleright$\ }{}
\SetCommentSty{itshape}
\newcommand{\UNE}{{\textsc{une}}}
\newcommand{\UW}{{\textsc{uw}}}
\newcommand{\USE}{{\textsc{use}}}
\newcommand{\N}{{\textsc{n}}}
\newcommand{\NW}{{\textsc{nw}}}
\newcommand{\SW}{{\textsc{sw}}}
\renewcommand{\S}{{\textsc{s}}}
\newcommand{\SE}{{\textsc{se}}}
\newcommand{\NE}{{\textsc{ne}}}
\newcommand{\DNW}{{\textsc{dnw}}}
\newcommand{\DSW}{{\textsc{dsw}}}
\newcommand{\DE}{{\textsc{de}}}
\newcommand{\X}{{\textsc{x}}}

\newcommand{\tri}{G_{\!\triangle}}
\newcommand{\surf}{G_{\!\lozenge}}
\newcommand{\suc}{\textsc{head}}
\newcommand{\pred}{\textsc{tail}}
\newcommand{\seg}{\textsc{seg}}
\newcommand{\gen}{\textsc{gen}}

\newcommand{\obj}{\theta}
\newcommand{\links}{\textsc{li}}
\newcommand{\lp}{\tau}
\newcommand{\bp}{\sigma}
\newcommand{\start}{s^0}
\newcommand{\occ}{\mathcal{T}}

\newcommand{\emp}{\mathcal{E}}
\newcommand{\anchor}{{a}}

\newcommand{\bo}{B_{\!\occ{}}}
\newcommand{\be}{B_{\!\emp{}}}
\newcommand{\beo}{B_{\!\emp{}^0}}

\newcommand{\coat}{{\textsc{Coat}}}
\newcommand{\init}{{\textsc{Init}}}
\newcommand{\fetch}{{\textsc{Fetch}}}
\newcommand{\noCheck}{{\textsc{skip}}}
\newcommand{\rhr}{\textsc{r}}
\newcommand{\lhr}{\textsc{l}}

\renewcommand{\O}{\mathcal{O}}

\title{Universal Coating by 3D Hybrid Programmable Matter}

\author{Irina Kostitsyna}{KBR at NASA Ames Research Center}{irina.kostitsyna@nasa.gov}{https://orcid.org/0000-0003-0544-2257}{}
\author{David Liedtke}{Paderborn University, Germany}{liedtke@mail.upb.de}{https://orcid.org/0000-0002-4066-0033}{}
\author{Christian Scheideler}{Paderborn University, Germany}{scheideler@upb.de}{https://orcid.org/0000-0002-5278-528X}{}

\authorrunning{I. Kostitsyna, D. Liedtke, C. Scheideler} 
\Copyright{Irina Kostitsyna and David Liedtke and Christian Scheideler}

\ccsdesc[500]{Theory of computation~Design and analysis of algorithms} 

\keywords{Programmable Matter, Coating, Finite Automaton, 3D Model}

\category{} 

\relatedversion{} 

\funding{This work was supported by the DFG Project SCHE 1592/10-1.} 

\nolinenumbers 

\EventEditors{John Q. Open and Joan R. Access}
\EventNoEds{2}
\EventLongTitle{42nd Conference on Very Important Topics (CVIT 2016)}
\EventShortTitle{CVIT 2016}
\EventAcronym{CVIT}
\EventYear{2016}
\EventDate{December 24--27, 2016}
\EventLocation{Little Whinging, United Kingdom}
\EventLogo{}
\SeriesVolume{42}
\ArticleNo{23}

\begin{document}
    
\maketitle              

\begin{abstract}
    Motivated by the prospect of nano-robots that assist human physiological functions at the nanoscale, we investigate the coating problem in the three-dimensional model for hybrid programmable matter. In this model, a single agent with strictly limited viewing range and the computational capability of a deterministic finite automaton can act on passive tiles by picking up a tile, moving, and placing it at some spot. The goal of the coating problem is to fill each node of some surface graph of size $n$ with a tile. We first solve the problem on a restricted class of graphs with a single tile type, and then use constantly many tile types to encode this graph in certain surface graphs capturing the surface of 3D objects. Our algorithm requires $\O(n^2)$ steps, which is worst-case optimal compared to an agent with global knowledge and no memory restrictions.
\end{abstract}
\section{Introduction}
\label{sec:introduction}
Recent advances in the field of molecular engineering gave rise to a series of computing DNA robots that are capable of performing simple tasks on the nano scale, including the transportation of cargo, communication, movement on the surface of membranes, and pathfinding~\cite{dnaCargo, dnaCargoCollective, dnaWalkMembrane, dnaPathFinding}. 
These results foreshadow future technologies in which a collective of computing particles cooperatively act as programmable matter - a homogenous material that changes its shape and physical properties in a programmable fashion.
Robots may be deployed in the human body as part of a medical treatment:
they may repair tissues by covering wounds with proteins or apply layers of lipids to isolate pathogens.
The common thread uniting these applications is the \emph{coating problem}, in which a thin layer of some specific substance is applied to the surface of a given object.

In the past decades, a variety of models for programmable matter has been proposed, primarily distinguished between passive and active systems.
In passive systems, particles 
move and bond to each other solely by external stimuli, e.g., current or light, or by their structural properties, e.g., specific glues on the sides of the particle.
Prominent examples are the DNA tile assembly models aTAM, kTAM and 2HAM (see survey in~\cite{tileAssemblySurvey}).
In contrast, particles in active systems solve tasks by performing computation and movement on their own.
Noteworthy examples include the Amoebot model, modular self-reconfigurable robots and swarm robotics~\cite{amoebotAnnouncement, MSR, MSR2, swarmRobotics}.
While computing DNA robots are difficult to manufacture, passive tiles can be folded from DNA strands efficiently in large quantities~\cite{dnaTileSurvey}.
A trade-off between feasibility and utility is offered by the hybrid model for programmable matter~\cite{hybridShapeFormation, hybridShapeRecognition, hybrid3D}, in which a single active agent that acts as a deterministic finite automaton operates on a large set of passive particles (called tiles) serving as building blocks.
A key advantage of the hybrid approach lies in the reusability of the agent upon completing a task.
While agents in purely active systems are expended as they become part of the formed structure, hybrid agents can be deployed again.
This property is beneficial in scenarios requiring the coating of numerous objects, such as isolating malicious cells within the human body.
Coating multiple objects concurrently, with each being individually coated by a single agent, allows for efficient~pipelining.

In the 3D hybrid model, we consider tiles of the shape of rhombic dodecahedra, i.e., polyhedra with 12 congruent rhombic faces, positioned at nodes of the adjacency graph of face-centered cubic (FCC) stacked spheres (see \cref{fig:modelgraph,fig:modeldirections}).
In contrast to rectangular graphs (e.g., \cite{cadbots}), this allows the agent to fully revolve around tiles without losing connectivity, which prevents the agent and tiles from drifting apart, e.g., in liquid or low gravity environments.
In this paper, we investigate the coating problem in the 3D hybrid model, in which the goal is to completely cover the surface of some impassable object with tiles, where tiles can be gathered from a material depot somewhere on the object's~surface.

\subsection{Our Results}
\label{subsec:contribution}

We present a generalized algorithm that solves the coating problem assuming that the agent operates on a graph $\tri{}$ of size $n$ and degree $\Delta \leq 6$ that is a triangulation of a closed 3D surface.
We assume a fixed embedding of $\tri{}$ in $\mathbb{R}^3$ in which edges have constantly many possible orientations, and that the boundary of each node in $\tri{}$ is a chordless cycle.
Our algorithm requires only a single type of passive tiles and solves the coating problem in $\O(n^2)$ steps, which is worst-case optimal compared to an algorithm for an agent with global knowledge and no restriction on its memory or the number of tile types. 
In the 3D hybrid model, the surface graph arises as the subgraph induced by nodes adjacent to a given object (some subset of nodes) where we assume holes in the object to be sufficiently large.
While that subgraph is not necessarily a triangulation with the properties described above, we show that our algorithm can be emulated on a restricted class of objects with a single type of tiles.
To realize the algorithm outside of that class, we construct a virtual surface graph on which our algorithm can be emulated in $\O(\Delta^2n^2)$ steps using $2^{2\Delta}$ types of passive tiles.
Notably, $\Delta$ is a constant in the 3D hybrid model.


\subsection{Related Work}
\label{subsec:relatedwork}

In recent years much work on the 2D version of the hybrid model has been carried out, yet the only publication that considers the 3D variant is a workshop paper from EuroCG 2020 in which an arbitrary configuration of $n$ tiles is rearranged into a line in $\O(n^3)$ steps \cite{hybrid3D}.
2D shape formation was studied in~\cite{hybridShapeFormation}; the authors provide algorithms that build an equilateral triangle in $\O(nD)$ steps where $D$ is the diameter of the initial configuration.
The problem of recognizing parallelograms of a specific height to length ratio was studied in~\cite{hybridShapeRecognition}.
The most recent publication \cite{hybridLine} solves the problem of maintaining a line of tiles in presence of multiple agents and dynamic failures of the tiles.

Closely related to the hybrid model is the well established Amoebot model, in which computing particles move on the infinite triangular lattice via a series of expansions and contractions.
In this model, a variety of problems was researched in the last years, including convex hull formation~\cite{amoebotConvexity}, shape formation~\cite{amoebotShape1, amoebotShape2}, and leader election~\cite{amoebotLeader}.
A~recent extension considers additional circuits on top of the Amoebot structure which results in a significant speedup for fundamental problems~\cite{amoebotCircuits}.
In~\cite{amoebotCoating2, amoebotCoating}, the authors solve the coating problem in the 2D Amoebot model;
in their variant, the objective is to apply multiple layers of coating to the object. 
In \cite{amoebot3Dcoating}, the authors solve the coating problem in the 3D Amoebot model.
In their approach, the object's surface is greedily flooded by amoebots that remain connected in a tree structure.
The process terminates for each amoebot when there are no more surrounding empty positions to move to.
Notably, given our agent's requirement to retrieve a tile after each placement, their approach cannot be applied or transferred to the hybrid model.

Deterministic finite automata that navigate the infinite 2D grid graph $\mathbb{Z}^2$ are considered in \cite{Czyzowicz2021NestFormation,Kant2021FortFormation}.
The authors of \cite{Czyzowicz2021NestFormation} address the challenge of building a compact structure, termed a nest, using available tiles.
They present a solution with a time complexity of $O(s^n)$, where $s$ denotes the initial span, and $n$ denotes the number of tiles.
In \cite{Kant2021FortFormation}, the authors focus on constructing a fort, a shape with minimal span and maximum covered area, in $\mathcal{O}(n^2)$ time.

In the field of modular reconfigurable robots, coating is often part of the shape formation problem.
In the 3D Catom model, a module of robots first assembles into a scaffolding~\cite{claytronicsScaffold} that is then coated by another module of robots~\cite{claytronicsCoating}.
The robots have spherical shape and reside in the FCC lattice;
in contrast to the hybrid model, they assume more powerful computation, sensing and communication capabilities.
The problem of leader election and local identifier assignment by generic agents in the FCC lattice is considered in~\cite{3dLeaderElection}.
Coating is approached differently in the field of swarm robotics where robots form a non-uniform spatial distribution around objects that are too heavy to be lifted alone~\cite{swarmRobotics}.

\section{Model and Problem Statement}
\label{sec:model}

\begin{figure}[t] 
    \centering
    \begin{minipage}{.49\textwidth}
        \centering
        \includegraphics[height=.65\linewidth,trim={{0} {12} {0} {24}},clip]{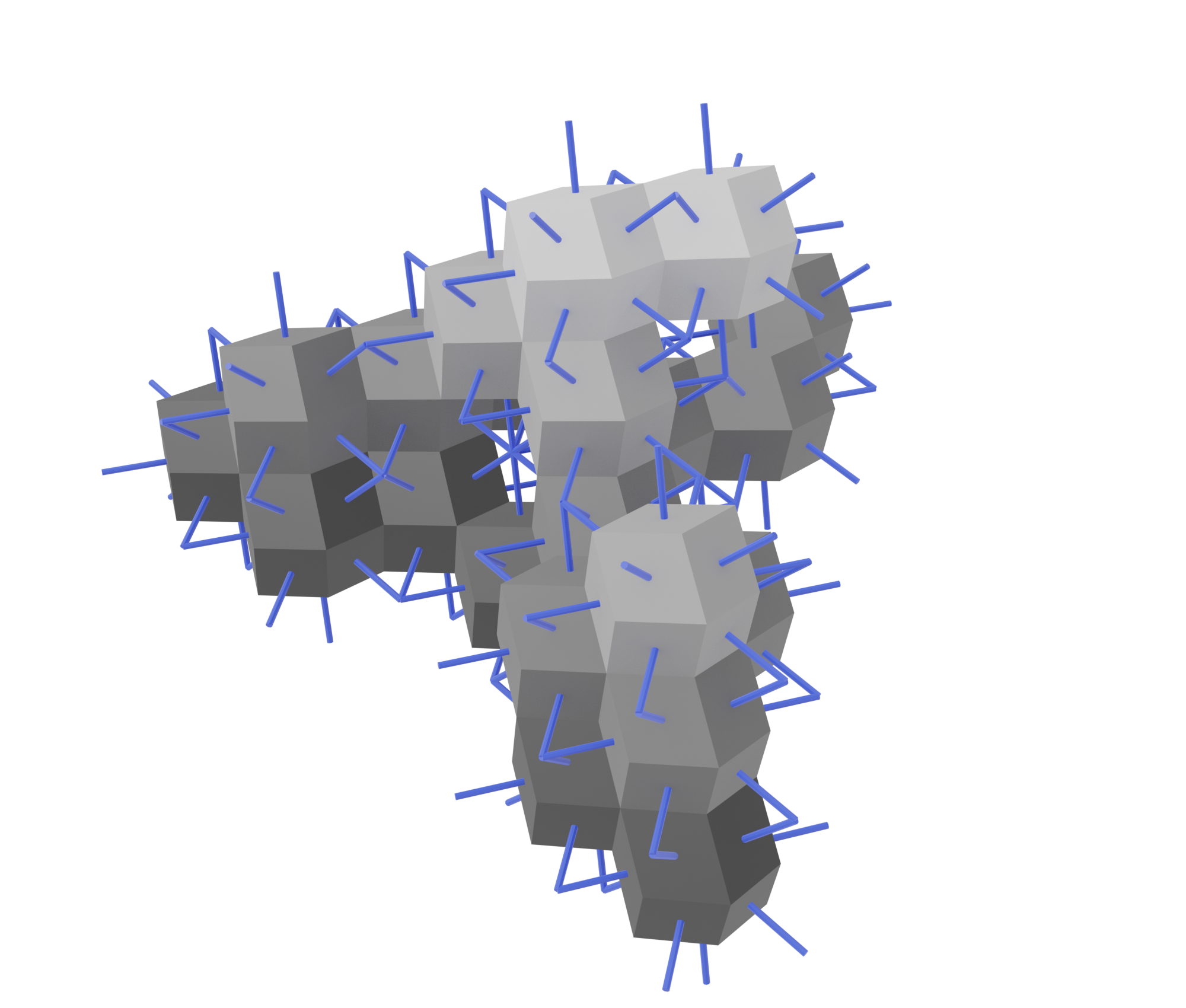}
        \caption{Tiled nodes of the underlying graph $G$ and their incident edges to empty nodes.}
        \label{fig:modelgraph}
    \end{minipage}%
    \hfill
    \begin{minipage}{.49\textwidth}
        \centering
        \includegraphics[height=.65\linewidth,trim={{0} {12} {0} {24}},clip]{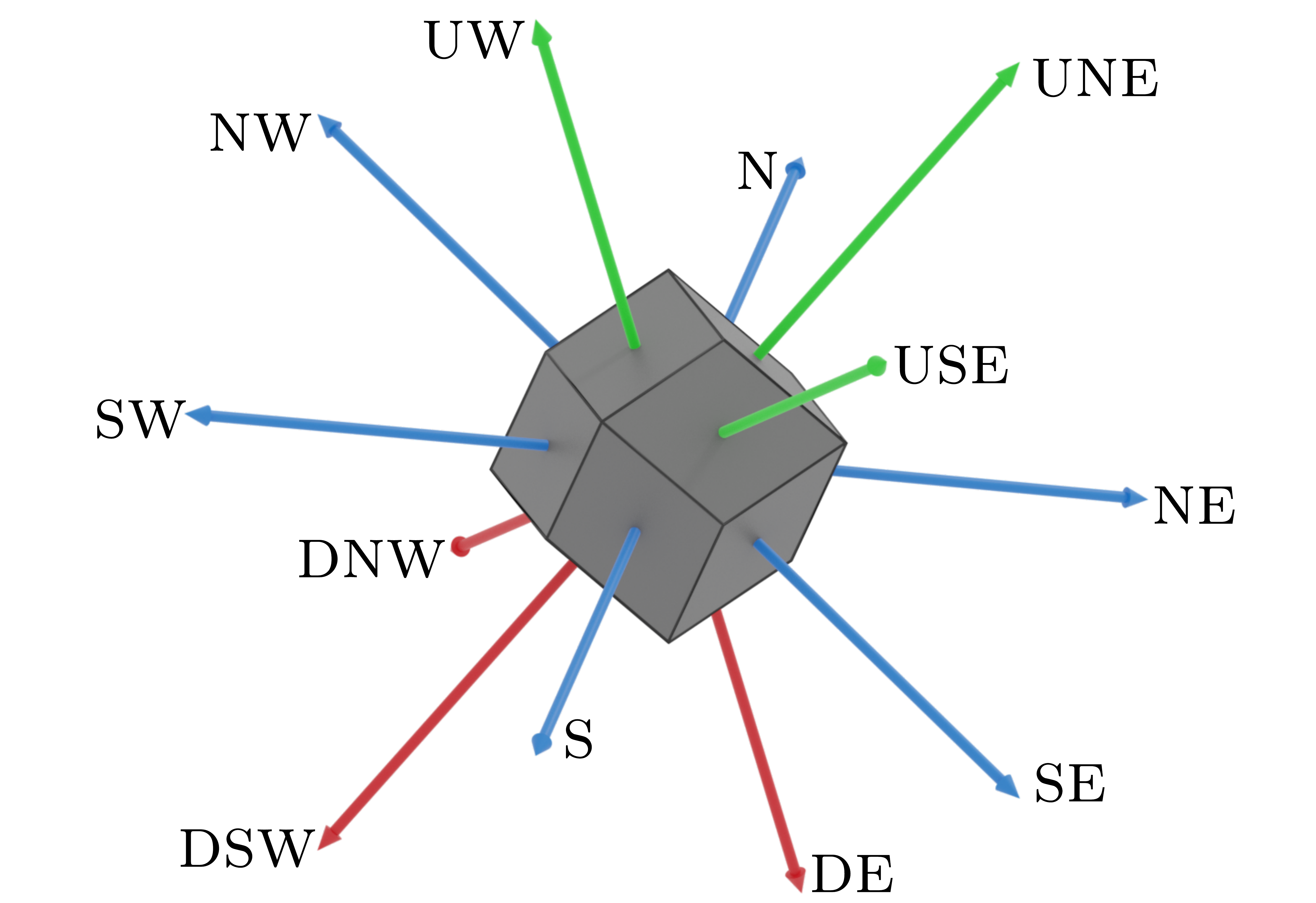}
        \caption{A passive tile (rhombic dodecahedron) and the twelve compass directions.}
        \label{fig:modeldirections}
    \end{minipage}
\end{figure}
In the \emph{3D hybrid model}, we consider a single active agent with limited sensing and computational power that operates on a finite set of passive \emph{tiles} positioned at nodes of some underlying graph $G$ that has a fixed embedding in $\mathbb{R}^3$.
We define the graph and embedding in \cref{subsec:graph}, the agent model in \cref{subsec:agent}, and the coating problem in \cref{subsec:problem}.

\subsection{Underlying Graph}
\label{subsec:graph}

Let $G = (V,E)$ be the adjacency graph of equally sized closely packed spheres at each point of the infinite face-centered cubic lattice (see \cref{fig:modelgraph}).
This graph can be embedded in $\mathbb{R}^3$ such that nodes have alternating cubic coordinates, i.e., each node $v$ has a coordinate $\vec{v} = (x,y,z)$ with $x,y,z \in \mathbb{Z}$ and $x+y+z$ is even.
Each node has twelve neighbors whose relative positions are described by the directions $\UNE$, $\UW$, $\USE$, $\N$, $\NW$, $\SW$, $\S$, $\SE$, $\NE$, $\DNW$, $\DSW$ and~$\DE$, which correspond to the following vectors in~the~embedding~of~$G$:
\begin{equation*} 
    \label{eq:embedding}
    \begin{array}{llllllllll}
        \vec{\UNE} & =  (1,1,0),   &  \vec{\UW}  & =  (0,1,1),   & \vec{\USE} & =  (1,0,1),   & \vec{\N}   & =  (0,1,-1),\\
        \vec{\NW}  & =  (-1,1,0),  &  \vec{\SW}  & =  (-1,0,1),  & \vec{\S}   & =  (0,-1,1),  & \vec{\SE}  & =  (1,-1,0),  \\
        \vec{\NE}  & =  (1,0,-1),  &  \vec{\DNW} & =  (-1,0,-1), & \vec{\DSW} & =  (-1,-1,0), & \vec{\DE}  & =  (0,-1,-1).
    \end{array}
\end{equation*}

Cells in the dual graph of $G$ w.r.t.\ the above embedding have the shape of rhombic dodecahedra, i.e., polyhedra with 12 congruent rhombic faces (see \cref{fig:modeldirections}).
This is also the shape of every cell in the Voronoi tessellation of $G$, i.e., that shape completely tessellates 3D space.
Consistent to the embedding, we denote by $v + \X$ the node $w$ that is neighboring $v$ in some direction $\X$, i.e., $\vec{w} = \vec{v} + \vec{\X}$.
Consider a finite set of tiles of $k$ distinguishable types (until \cref{sec:construction} we only consider $k=1$).
Tiles have the shape of rhombic dodecahedra and are passive, in the sense that they cannot perform computation or movement on their own.
A node $v \in V$ is either \emph{tiled}, if there is a tile positioned at $v$, or \emph{empty}, otherwise.
Except for the material depot (which we introduce in \cref{subsec:problem}), nodes can hold at most one tile at a time.
We denote by $\occ{}$ the set of tiled nodes, and by $\emp{}$ the (infinite) set of empty nodes.

\subsection{Agent Model}
\label{subsec:agent}

The agent $r$ is the only active entity in this model.
It can place and remove tiles of any type at nodes of $G$ and loses and gains a unit of material in the process.
We assume that it initially carries no material and that it can carry at most one unit of material at any time.
The agent has the computational capabilities of a deterministic finite automaton performing \emph{Look-Compute-Move} cycles.
In the \emph{look}-phase, it observes tiles at its current position $p$ and the twelve neighbors of $p$, and if there are tiles, it observes their types as well.
The agent is equipped with a compass that allows it to distinguish the relative positioning of its neighbors.
Its initial rotation and chirality can be arbitrary, but we assume that it remains consistent throughout the execution.
In the \emph{compute}-phase the agent determines its next state transition according to the finite automaton.
In the \emph{move}-phase, the agent executes an \emph{action} that corresponds to that state transition.
It either (i) moves to an empty or tiled node adjacent to $p$, (ii) places a tile (of any type) at $p$, if $p \notin \occ{}$ and $r$ carries material (we call that \emph{tiling} node $p$), (iii) removes a tile from $p$, if $p\in \occ{}$ and $r$ carries no material, (iv) changes the tile type at $p$, or (v) terminates.
During (ii) and (iii), the agent loses and gains one unit of material, respectively.
While the agent is technically a finite automaton, we describe algorithms from a higher level of abstraction textually and through pseudocode using a constant number~of variables of constant size domain.

\subsection{Problem Statement}
\label{subsec:problem}


Denote by $G(W)$ the subgraph of $G$ induced by some set of nodes $W \subseteq V$, by $d(v,w)$ the distance (length of the shortest path) between nodes $v,w \in V$ w.r.t. $G$, and by $d_W(v,w)$ the distance w.r.t. $G(W)$.
Consider a connected subset $\obj{}\subset V$ of impassable and static nodes, called \emph{object}.
Any node is either an object node, empty or tiled such that $\obj{}$ and the sets of empty nodes $\emp{}$ and tiled nodes $\occ{}$ are pairwise disjoint.
A \emph{configuration} is the tuple $C = (\occ{},\obj, p)$, and we call $C$ \emph{valid}, if $G(\occ{} \cup \obj{} \cup \{p\})$ is connected.
We assume that holes in the object have width larger than one, i.e., $d_\obj{}(v,w) \leq 2$ for any $v,w \in \obj$ with $d(v,w) \leq 2$.

Let $C^0 = (\occ{}^0,\obj{}, p^0)$ be a valid initial configuration with $\occ{}^0=\{p^0\}$.
Superscripts~generally refer to step numbers and may be omitted if they are clear from the context.
Define the \emph{coating layer} as the maximum subset $L\subset V \setminus \obj{}$ such that for each node $v \in L$ there is an object node $w \in \obj$ with $d(v,w) = 1$ and $d_L(v,p^0) < \infty$ (w.r.t. $G(L)$).
The latter condition excludes unreachable nodes that are separated by the object, e.g., the inner surface of a hollow sphere.
We assume a \emph{material depot} at the agent's initial position $p^0$, that is $p^0$ is a node with the special property of holding at least $|L|$ units of material.
An algorithm solves the \emph{coating problem}, if its execution results in a sequence of valid configurations $C^0,\dots,C^{t^*}$ such that $\occ{}^{t^*} = L$, $C^{t}$ results from $C^{t-1}$ for $1\leq t \leq {t^*}$ by applying some action (i)--(iv) to $p^{t-1}$, and the agent terminates (v) in step ${t^*}$.
\section{The Coating Algorithm}
\label{sec:algorithm}

In this section, we give a generalized coating algorithm for a surface graph $\tri$ whose node set is the coating layer $L$.
The agent operates only in $\tri$ and all notation in this section is exclusively w.r.t. $\tri$.
We assume that (1)~$\tri$ is a triangulation of a closed 3D surface (e.g., \cref{fig:idea1}) with an embedding in $\mathbb{R}^3$ in which edges have constantly many possible orientations.
Define the $i$-\emph{neighborhood} $N_i(W)$ of $W$ as the set of nodes $v \in L$ with $d(v,w) \leq i$ for some node $w \in W$, and the \emph{boundary} as $B(W) \coloneqq N_1(W) \setminus W$.
We write $N_i(w)$ and $B(w)$ for $W = \{w\}$, and use subscripts to denote subsets of only empty or only tiled nodes, e.g., $\be(w) = B(w) \cap \emp{}$.
Further, we assume that (2) $\tri(B(v))$ contains precisely one simple cycle for any $v \in L$ (we say that $B(v)$ is \emph{chordless}).
Notably, this assumption does not weaken our results, since there are surface graphs in which properties (1) and (2) naturally arise, and otherwise we can emulate the properties using additional tile types (see \cref{sec:construction}).

\begin{figure}[t]
    \centering
    \begin{subfigure}[c]{0.49\linewidth}
        \includegraphics[width=\linewidth]{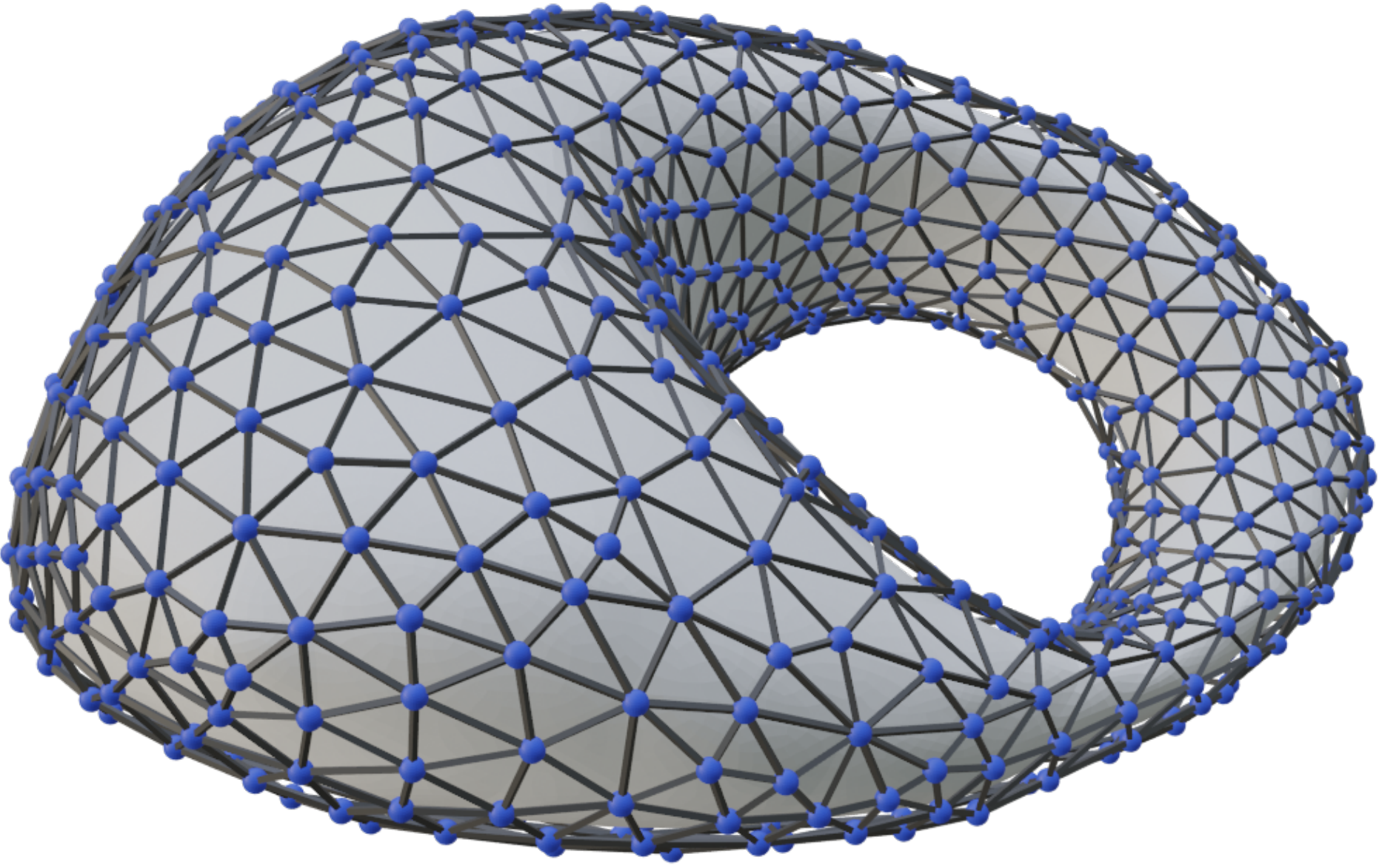} 
        \subcaption{}
        \label{fig:idea1}
    \end{subfigure}\hfill
    \begin{subfigure}[c]{0.49\linewidth}
        \includegraphics[width=\linewidth]{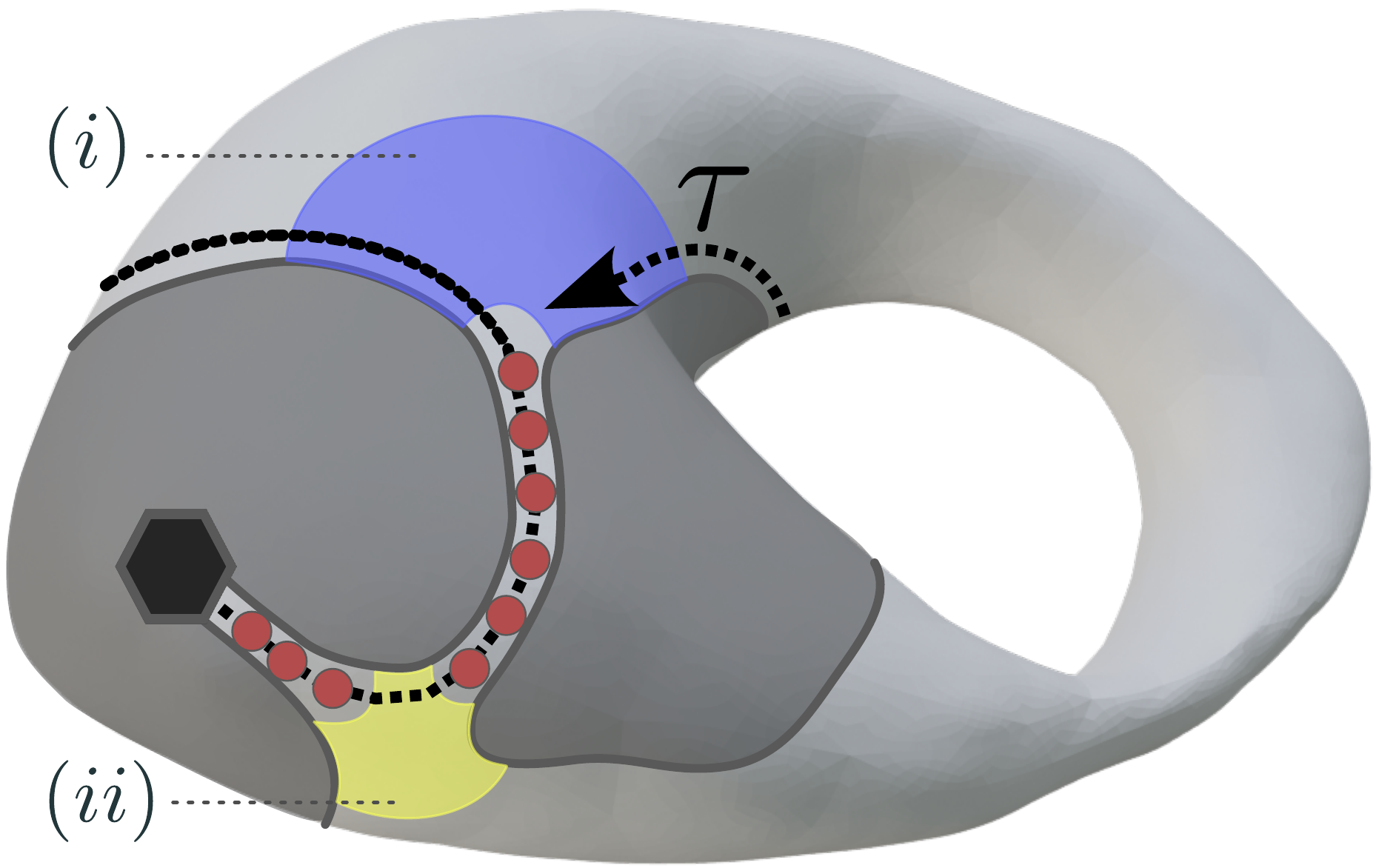} 
        \subcaption{}
        \label{fig:idea2}
    \end{subfigure}%
    \caption{(a) Example of a triangulation $\tri{}$. (b) The simple path $\tau$ that starts at the material depot (hexagon) along the boundary of tiled nodes (opaque surface area). Links are depicted as circles. The blue area ($i$) is the \emph{range} explored by the agent. In the yellow area ($ii$), $\tau$ is `exposed' to nodes that are not tiled and not links. The agent explores the blue area to find an overlap with the yellow area, in which case a link can be tiled while preserving connectivity of $\tri(\emp)$.}
    \label{fig:idea}
\end{figure}


\subsection{High-Level Description and Preliminaries}
\label{subsec:preliminaries}

The main challenge is the exploration of $\tri{}$ to find an empty node that can be tiled while maintaining a path back to the material depot, as it is the only source of additional tiles.
The exploration of all graphs of diameter $D$ and degree $\Delta$ requires $\Omega(D \log(\Delta))$ memory bits \cite{graphExploration}.
With constant memory already the exploration of plane labyrinths (grid graphs in $\mathbb{Z}^2$) requires two pebbles (placable markers) \cite{exploration1pebble,exploration2pebbles}.
Our agent has only constant memory and while it is not provided with pebbles, we can use tiles to aid the exploration of $\tri{}$.
Since at some point all nodes in $\tri{}$ must be tiled, and our algorithm uses only a single type of tiles, i.e., tiles are indistinguishable, we cannot directly use tiles as markers.
Instead, our strategy involves strategically placing tiles to create a 'narrow tunnel' of empty nodes leading to the tile depot.
We then employ the left- and right-hand-rule (LHR, RHR) commonly used in labyrinth traversal to navigate through the resulting tile structure.
In this process, the agent consistently maintains contact with the set of tiled nodes, either on its left or right.
Essentially, we `mark' empty nodes by disconnecting tiles in their boundary such that they become part of the tunnel.
This concept is formalized through the introduction of \emph{links}:

\begin{definition}
    \label{def:1link}
    A node $v \in \emp{}$ is a \emph{link}, if $\be(v)$ is disconnected.
    A~node $v \in \emp{}$ \emph{generates} (a link at) node $w$, if tiling node $v$ turns $w$ into a link.
    Similarly, $v$ \emph{consumes} (a link at) node $w$, if by tiling $v$, $w$ is no longer a link.
\end{definition}

From a high level perspective, the agent traverses the boundary of tiled nodes by following the LHR until it finds a node to place its carried tile at.
It then moves back to the material depot following the RHR, gathers material and repeats the process.
Following the LHR and RHR, the agent cannot leave the connected component of $\tri(\emp)$ that contains its position.
Hence, it is crucial that each tile placement preserves connectivity of $\tri(\emp)$, as otherwise the agent might never find its way back to the material depot, or it might terminate without tiling some node in $L$.
A naive approach to maintain connectivity would be to never place a tile at a link.
However, that strategy only works if the surface that is captured by $\tri{}$ is simply connected, i.e., it does not contain any hole.
In fact, the link property is necessary for some node $v \in \emp{}$ to be a cut node w.r.t. $\tri(\emp)$ but it is not sufficient.
If the surface contains holes, then the naive approach would converge $\tri(\emp{})$ to a cyclic graph containing only links, which again may be impossible to explore.
Hence, links must be tiled eventually.

In our algorithm, we build a simple path $\tau$ of empty nodes that contains all links and is completely contained in the boundary $B(\occ{})$ of tiled nodes (see \cref{fig:idea2}).
We cannot ensure connectivity of links on $\tau$, i.e., there can be multiple sections at which $\tau$ is `exposed' to empty nodes that are not links.
However, we can ensure that links on $\tau$ are `sufficiently close' to each other which we will formalize below by the notion of \emph{segments}.
In each traversal of $\tau$, the agent visits all links precisely once by following the LHR until there are no more links in the segment at the agent's position.
Afterwards, it explores a constant size neighborhood called \emph{range}, where we define ranges such that they contain all `close' links (all segments in some local neighborhood).
The presence of a link in that range indicates that we have found one of the above mentioned `exposures' of $\tau$ that was already visited before, i.e., there exists a link in some cycle of $\tri{}(\emp{})$ that can safely be tiled.
Note that the agent does not necessarily traverse $\tau$ fully.
As an example, the subpath of $\tau$ that follows the topmost link $v$ in \cref{fig:idea2} is not traversed, since the next segment following node $v$ contains no link.

\paragraph*{Additional Terminology.}
\begin{figure}[!t]
    \centering
    \begin{subfigure}[c]{0.33\linewidth}
        \includegraphics[width=\linewidth]{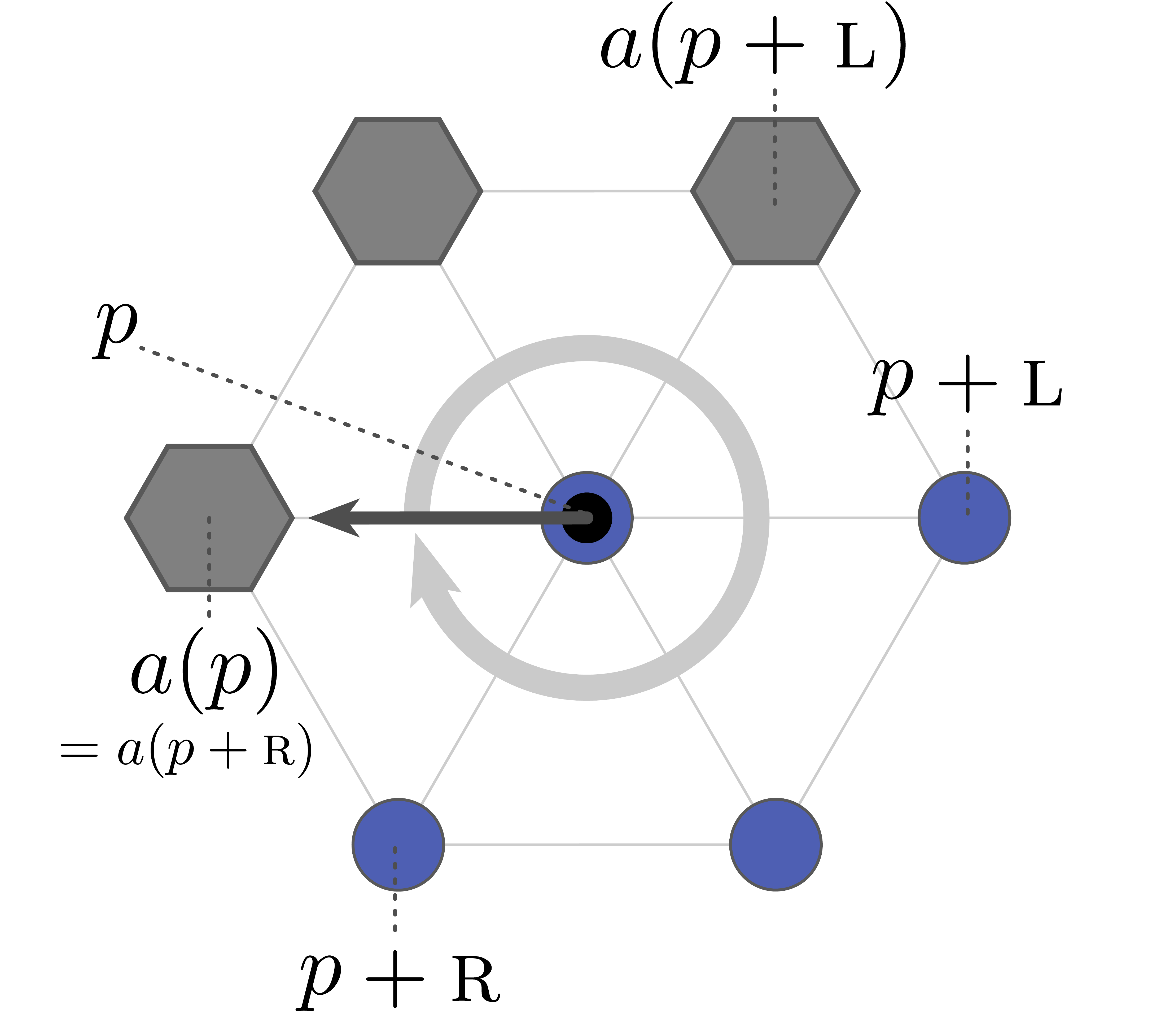}
        \subcaption{}
        \label{fig:anchor}
    \end{subfigure}%
    \begin{subfigure}[c]{0.33\linewidth}
        \includegraphics[width=\linewidth]{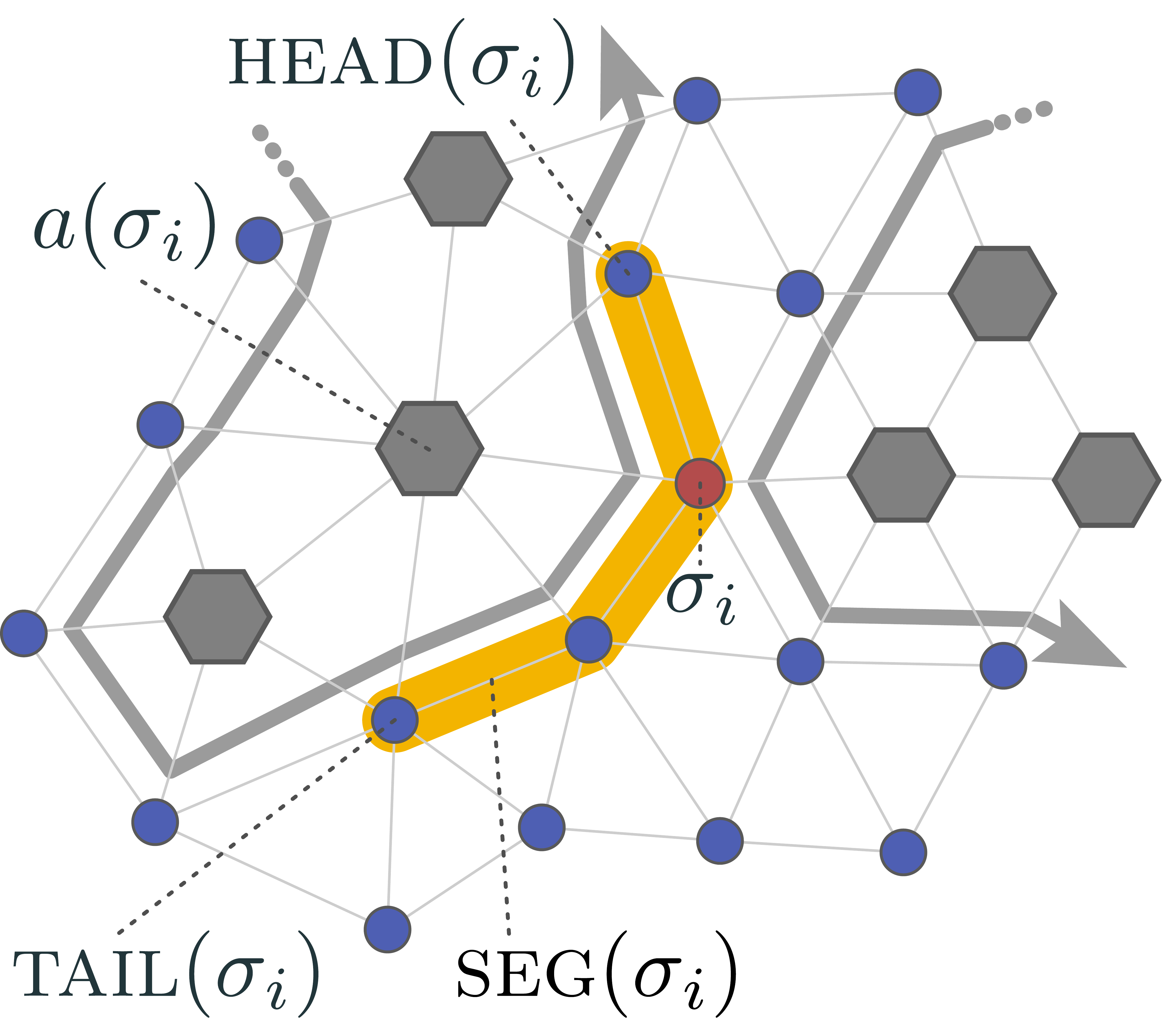}
        \subcaption{}
        \label{fig:segment}
    \end{subfigure}%
    \begin{subfigure}[c]{0.33\linewidth}
        \includegraphics[width=\linewidth]{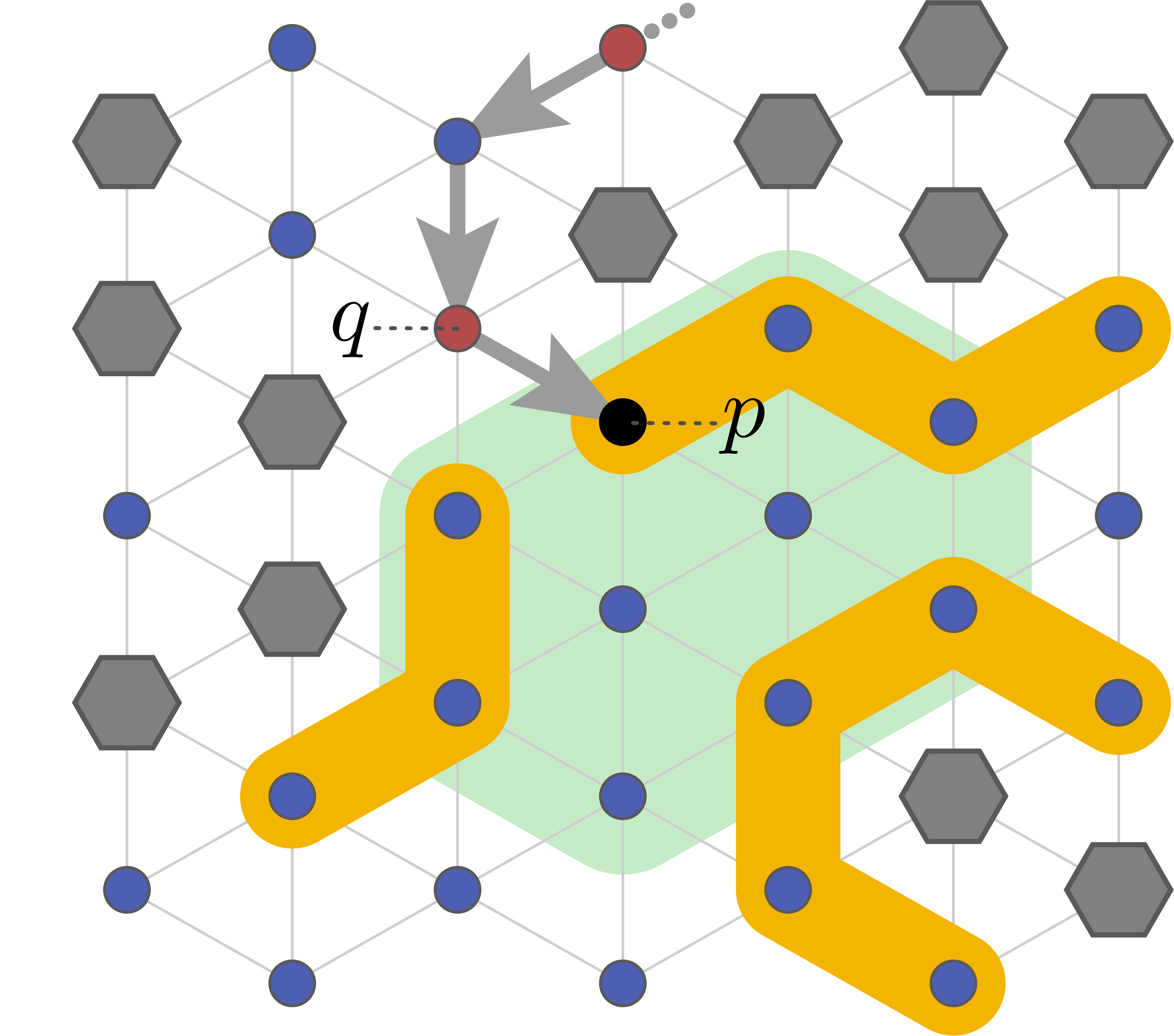}
        \subcaption{}
        \label{fig:range}
    \end{subfigure}%
    \caption{Tiled nodes are depicted as hexagons, empty nodes as disks, links are depicted in red (node $\sigma_i$ in (b), $q$ in (c) and the uppermost central node in (c)). (a) Example of the updates made to the agent's anchor $a(p)$ while following the LHR and RHR. (b) $\seg(\bp_i), \pred(\bp_i)$ and $\suc(\bp_i)$ for some node $\bp_i \in \bp$ in an example configuration. (c) Example of the $i$-range $R_i(p,q)$ for $i = 2$ (orange, opaque) together with $N_2(p)$ w.r.t. $\tri(\emp{} \setminus \{q\})$ (green, transparent).}
\end{figure}

The agent (at node $p$) maintains a direction pointer to some tiled node $\anchor(p) \in B(p)$, called \emph{anchor} of $p$.
Using the analogy of wall-following in a labyrinth, $\anchor(p)$ is the node at which the agent's `hand' is currently placed.
Initially, the agent leaves the tile depot $p^0$ to an arbitrary neighbor and sets its anchor to $p^0$.
Afterwards, it can follow the LHR (RHR) by moving to the first empty node $v$ in a clockwise (counter-clockwise) order of $B(p)$ starting at $a(p)$ and sets its anchor to the last tiled node between $a(p)$ and $v$ in that order (see \cref{fig:anchor}).
Note that the agent's anchor does not necessarily change after each move.
Subsequently, $\lhr{}$ and $\rhr{}$ denote the direction to the next node according to the LHR and RHR, respectively.

Let $\start \in B(p^0)$ be a dedicated starting node chosen arbitrarily in the initialization.
Define $\bp = (\bp_0,\dots,\bp_m)$ as the path along nodes of $B(\occ{})$ according to a LHR traversal that starts and ends at $\start = \bp_0 = \bp_m$. 
Note that $\bp$ is not necessarily a simple cycle.
In fact, in a full traversal of $\bp$, the agent can visit a link multiple times with its anchor set to different tiled nodes on each visit (e.g., $\bp_i$ in \cref{fig:segment}).
To avoid ambiguity, $\anchor(\bp_i)$ refers to the anchor \emph{after} moving from $\bp_{i-1}$ to $\bp_i$, and $\anchor(v)$ refers to the anchor at the first occurence of $v$ on $\bp$.

The \emph{segment} $\seg(\bp_i)$ of a node $\bp_i$ (see \cref{fig:segment}) contains all nodes that can be reached from $\bp_i$ by following the LHR or RHR while keeping the anchor fixed at $\anchor(\bp_i)$.
Simply put, $\seg(\bp_i)$ is the set of nodes that touch the anchor $a(\bp_i)$ from the same `side' as $\bp_i$.
Any node $\bp_j \in \seg(\bp_i)$ is a \emph{successor} of $\bp_i$, if $j > i$, or a \emph{predecessor} of $\bp_i$, if $j < i$.
As an example, node $\sigma_i$ in \cref{fig:segment} has two predecessors and one successor.
Denote by $\suc(\bp_i)$ the node without a successor in $\seg(\bp_i)$, and $\pred(\bp_i)$ the node without a predecessor in $\seg(\bp_i)$.


We can now formally introduce the above mentioned path $\lp$.
Define $\lp = (\lp_0,\dots,\lp_l)$ as a maximal simple sub-path of $\bp$ that starts at $\start$, i.e., $\lp_i = \bp_i$ for any $0 \leq i \leq l$.
The definition of $\seg(\cdot),\suc(\cdot)$ and $\pred(\cdot)$ directly carry over from $\bp$.
For simplicity, we write $v \in \lp$ or $v \in \bp$ if $\lp$ or $\bp$ contains node $v$.

Finally, we introduce the aforementioned \emph{range} that is explored by the agent before each tile placement.
Consider the agent to be positioned at some empty node $p$, such that the node $q = p + \rhr{}$ is a link.
The $i$-\emph{range} $R_i(p,q)$ (see \cref{fig:range}) is a specific neighborhood of empty nodes that is defined as if node $q$ were tiled.
Consider a node $v$ that can be reached from $p$ in at most $i$ steps without moving through $q$ or any tiled node.
If the segment $\seg(v)$ does not contain $q$, we add it to $R_i(p,q)$.
Otherwise, that segment is separated at $q$ and only the part that contains $v$ is added.
The formal definition is as follows: 

\begin{definition}
    \label{def:range}
    The $i$-\emph{range} of $p$ w.r.t. $q$ is the set of nodes $$R_i(p,q) \coloneqq \bigcup_{v \in N_i(p)}\bigcup_{\;\;w \in \bo(v)} \seg(v)$$
    where $N_i(p)$ and $\seg(v)$ are w.r.t. $\tri(\emp \setminus \!\{q\})$ and anchor $w$.
\end{definition}

\subsection{Algorithm Details and Pseudocode}
\label{subsec:algorithm}

The coating algorithm (see pseudocode in \cref{alg:algorithm}) consists of an initialization \init{} (lines 1--3) and two phases \coat{} (lines 4--10), dedicated to the tile placement, and \fetch{} (lines 11--16), dedicated to the gathering of material, the traversal of $\lp$, and the termination.
The agent switches between phases $\coat{}$ and $\fetch{}$ after each tile placement.
In the pseudocode, 
$\links$ and $\gen{}$ denote the set of all links and generators (see \cref{def:1link}). 

\begin{algorithm}[!b]
    \caption{Coating Algorithm}
    \label{alg:algorithm}
    \SetAlgoVlined
    \DontPrintSemicolon
    Phase \init{}:\\
    \nl gather material from $p^0$; move to an arbitrary $\start \in B(p^0)$\;
    \nl store the direction of $p^0$ w.r.t. $\start$; $\anchor(\start)\leftarrow p^0$; $\noCheck \leftarrow false$\;
    \nl move $\lhr{}$; enter phase \coat{} \Comment*[r]{$p \leftarrow \start{} + \lhr{}$}
    \BlankLine
    Phase \coat{}:\\
    \nl \uIf{$\noCheck{}$ \normalfont{or} $R_3(p, p + \rhr{}) \cap \left(\links \cup \{\start{}\}\right) = \emptyset$}{
        \nl place a tile at $p$; $\noCheck{} \leftarrow false$; enter phase \fetch{}\;
    }
    \nl \Else{
        \nl move \rhr{} \Comment*[r]{$p \leftarrow p + \rhr{}$}
        \nl \lIf(\Comment*[f]{$p \leftarrow p + \rhr{}$}){$p \in \gen$}{move \rhr{}}
        \nl \lIf{$p \in \gen$ \normalfont{and} $p + \rhr{} \notin \links \cup \{\start{}\}$}{$\noCheck{} \leftarrow true$}
        \nl place a tile at $p$; enter phase \fetch{}\;
    }
    \BlankLine
    Phase \fetch{}:\\
    \nl \lWhile(\Comment*[f]{$p \leftarrow p + \rhr{}$}){$p \neq \start$}{
        move \rhr
    }
    \nl move to $p^0$; gather material from $p^0$; move to $\start$\;
    \nl \lIf{$\be(\start{}) = \emptyset$}{
        place a tile at $\start$; terminate
    }
    \nl \While{$p \in \links \cup \{\start{}\}$ \normalfont{or} $v \in \links$ \normalfont{for a successor $v$ of $p$ in} $\seg(p)$}{
        \nl move \lhr \Comment*[r]{$p \leftarrow p + \lhr{}$}
    }
    \nl enter phase \coat{}\;
\end{algorithm}

In phase \init{} (lines 1--3), the agent gathers material and moves to an arbitrary node $\start\in B(p^0)$.
It stores the direction of $p^0$ w.r.t. $\start$ such that it can recognize $\start$ by the adjacent material depot at a later visit, and it initializes $\anchor{}(\start{}) \leftarrow p^0$ and $\noCheck \leftarrow false$.
Note that $\start$ is the first node of the paths $\bp$ and $\lp$, $\anchor(p)$ is the agent's anchor, and $\noCheck$ is a flag that indicates that the search for links in the $3$-range is skipped in the next execution of phase \coat{} (see line 4).
The flag is necessary to maintain a crucial invariant which we elaborate in the proof of \cref{lem:invariant}.
Afterwards, the agent moves \lhr{} and enters phase \coat{}.

Phase \coat{} is always entered such that $p \notin \links \cup \{\start{}\}$ and $p + \rhr{}$ is the last node $v\in\lp$ (i.e., with maximum index) for which $v \in \links \cup \{\start{}\}$.
In each execution of phase \coat{} the tile that is carried by the agent is either placed directly at $p$ (lines 4--5), or at some link that can be reached by following the RHR for at most two steps (lines 7--10).
In any case, the agent enters phase \fetch{} afterwards.
The position of the next tile depends on the following criteria:
If the flag $\noCheck{}$ is set to $true$, then the tile is placed at $p$ and $\noCheck{}$ is set to $false$ afterwards.
The tile is also placed at $p$, if $R_3(p,p+\rhr{})$ (see \cref{def:range}) does not contain any node of $\links \cup \{\start{}\}$.
Otherwise, the agent must have found some link or the starting node, which implies that $R_3(p,p+\rhr{})$ contains a node $v \in \lp$ with smaller index than $p \in \lp$, i.e., the agent has detected a cycle in $\tri{}(\emp)$ in which it can safely place a tile at some link.
It is crucial that no link is generated on the `wrong side' of the newly placed tile as this link may later lead to a false detection.
Let the agent's position w.r.t. $\bp$ be $p = \bp_i$, i.e., $\bp_{i-1},\bp_{i-2}$ and $\bp_{i-3}$ are the next three nodes visited by following the RHR (see \cref{fig:alg}).
If $\bp_{i-1} \notin \gen$, then node $\bp_{i-1}$ is tiled, otherwise node $\bp_{i-2}$ is tiled.
In the latter case, $\noCheck$ is set to $true$ whenever $\bp_{i-2} \in \gen$ and $\bp_{i-3}\notin \links$ (line 9).
The flag $\noCheck{}$ ensures that node $\bp_{i-3}$ is tiled next such that the agent again enters phase \coat{} $\lhr{}$ of the last node $v \in \lp$ with $v \in \links \cup \{\start{}\}$.

\begin{figure}[!t]
    \centering
    \begin{subfigure}[c]{0.32\linewidth}
        \includegraphics[width=\linewidth]{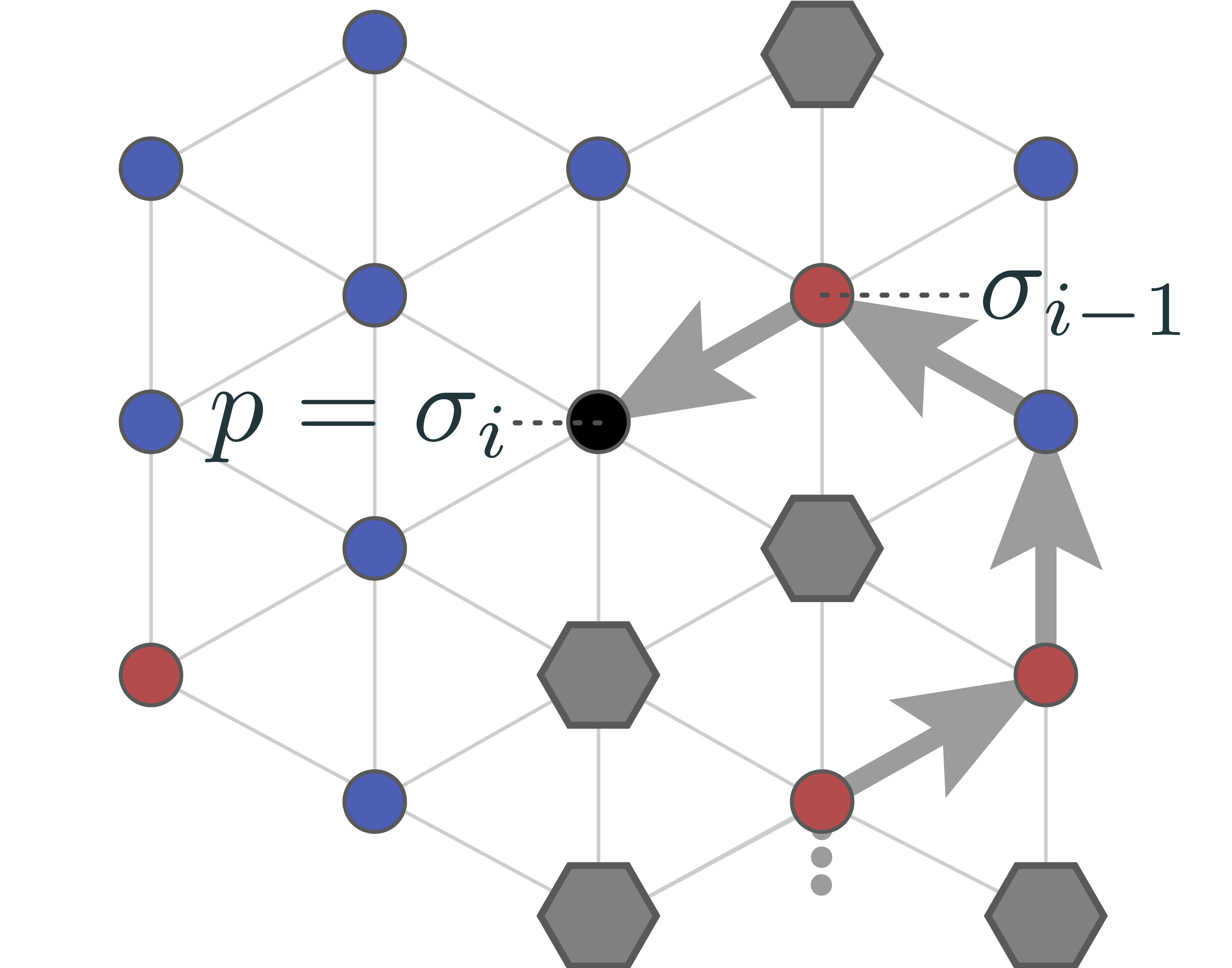}
        \subcaption{}
        \label{fig:alg1}
    \end{subfigure}
    \begin{subfigure}[c]{0.32\linewidth}
        \includegraphics[width=\linewidth]{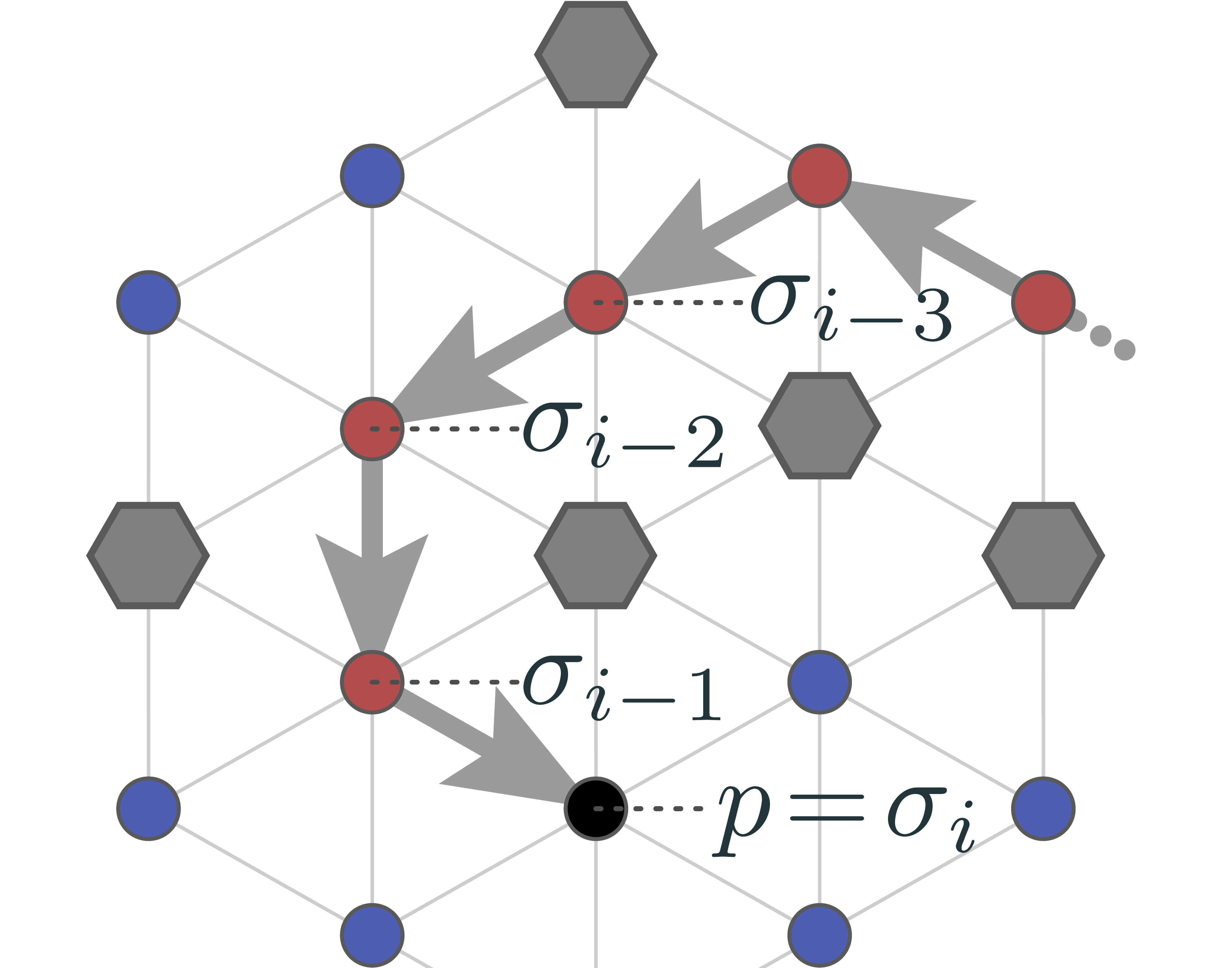}
        \subcaption{}
        \label{fig:alg2}
    \end{subfigure}
    \begin{subfigure}[c]{0.31\linewidth}
        \includegraphics[width=\linewidth]{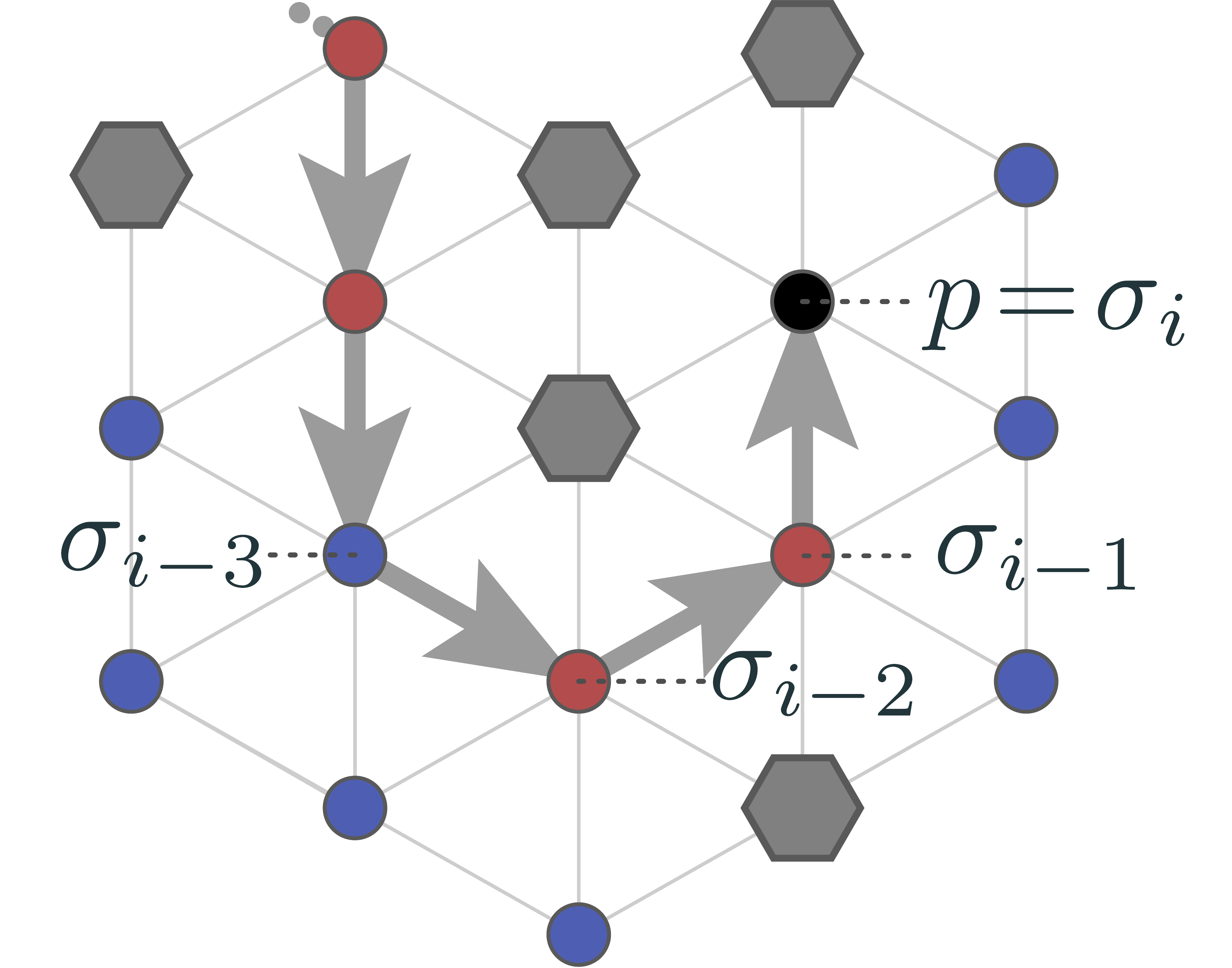}
        \subcaption{}
        \label{fig:alg3}
    \end{subfigure}%
    \caption{Examples in which the next tile is placed at some link (red disk) in phase \coat{}: $\sigma_{i-1}$ is tiled in (a); $\sigma_{i-2}$ is tiled in (b) and (c). Only in (c), $\noCheck$ is set to $true$ since $\sigma_{i-3} \notin \links \cup \{\start\}$. As a result, $\sigma_{i-3}$ is tiled on the next visit.}
    \label{fig:alg}
\end{figure}

In phase \fetch{}, the agent moves $\rhr{}$ until it is positioned at $\start$ (line 11), which it detects by the adjacent material depot and the direction stored in phase \init{}.
It moves to $p^0$, gathers material and returns to $\start$ (line 12).
If $\start$ has no empty neighbors, it places a tile at $\start$ and terminates (line 13).
Otherwise, the agent moves $\lhr{}$ as long as $p \in \links \cup \start$ or whenever a successor of $p$ in $\seg(p)$ is a link, and switches to phase \coat{} afterwards (lines 14--16).
In that step, the agent implicitly explores $\seg(p)$, since $p$ can have~multiple~successors.

\subsection{Analysis}
\label{subsec:analysis}

Consider an initial configuration $C^0 = (\occ^0,\obj{},p^0)$ with $p^0 \in \occ^0 \subseteq L$, and a material depot of size at least $|L| - |\occ^0| + 1$ at $p^0$.
In the problem statement we assume that $p^0$ is the only initially tiled node, but now we allow $\occ^0$ to contain multiple tiled nodes besides $p^0$.
This will later become useful in \cref{sec:construction} where we construct a virtual graph in which some nodes are tiled initially.
We analyze \cref{alg:algorithm} given that $C^0$ satisfies the following definition:

\begin{definition}
    \label{def:coatability}
    A configuration $C^0 = (\occ^0,\obj{},p^0)$ is coatable w.r.t. $\tri$, if $|B(v)| \leq 6$ for any $v \in \emp^0$, $\beo(p^0) \neq \emptyset$, $\links^0 = \emptyset$, and $\emp^0$ is connected.
\end{definition}

We aim to maintain five properties as invariants:
\emph{\textbf{P1}}: Links may only occur on the simple path $\lp$, i.e., $\links \subseteq \lp$.
\emph{\textbf{P2}}: All links are connected by a sequence of overlapping segments to the starting node, i.e., $\pred(v) \in \links \cup \{\start{}\}$ for any $v \in \links$.
\emph{\textbf{P3}}: The subpath of $\lp$ from $\start{}$ to the last link on $\lp$ induces no cycle in $\tri$, i.e., for any $i < j \leq k$ with $\lp_k \in \links$: if $d(\lp_i,\lp_j) = 1$, then $j = i+1$.
\emph{\textbf{P4}}: The boundary of any link contains precisely two connected components of empty nodes, i.e., $||\be(v)|| = 2$ for any $v \in \links$ where $||\cdot||$ denotes the number of connected components.
\emph{\textbf{P5}}: There exists a node of $\lp$ at which the agent enters phase \coat{}, i.e., either $\links = \emptyset$ or there is an $i$ such that $\lp_i \notin \links \cup \{\start{}\}$ and $\suc(\lp_i) = \lp_i$.


Observe that all properties hold initially by $\links^0 = \emptyset$.
The structure of our proof is as follows:
We prove termination given that $P5$ is maintained, and that $\emp$ never disconnects given that $P1$--$P4$ are maintained.
Since the agent always finds a node to place the next tile at by $P5$, there must eventually be a step in which $\be(\start{}) = \emptyset$.
Since $\emp{}$ remains connected until that step, $L = \occ{}$ holds after the last tile is placed at $\start{}$.
Finally, we show that $P1$--$P5$ are maintained as invariants.
We start with the termination of the algorithm:

\begin{lemma}
    \label{lem:termination}
    If $P5$ holds in step $t$ in which the agent gathers material at $p^0$, then there is a step $t^+ > t$ in which the agent enters $p^0$ again or terminates.
\end{lemma}

\begin{proof}
    Assume by contradiction that the agent does not terminate or enter $p^0$ again in any step $t' > t$.
    There are two cases: the agent places a tile in phase \coat{} and moves to a connected component of $\emp{}$ that does not contain $\start$, or the agent never enters phase \coat{}, i.e., it moves indefinitely $\lhr{}$ in phase \fetch{}.
    If the agent traverses a simple path from $\start{}$ to some node $v$ by moving $\lhr{}$, places a tile at $v$ and moves $\rhr{}$ afterwards, it must enter the connected component of $\emp{}$ that contains $\start{}$.
    Hence, in the first case, the agent places a tile at $v$ after visiting $v$ at least twice, i.e., it has fully traversed $\lp$, which contradicts the existence of node $\lp_i$ specified by $P5$.
    In the second case, the agent never reaches $\lp_i$, as it would otherwise enter phase \coat{} and place a tile.
    This implies a cycle $(\lp_j,...,\lp_k)$ with $\lp_{k+1} = \lp_{j}$ and $j < i$, which contradicts that $\lp$ is a simple path.
    Hence, there is a step $t^+$ in which $p^0$ is entered again or $\be(\start{}) = \emptyset$ and the agent terminates.
    \end{proof}

Subsequently, our notation refers to some step $t$ in which the agent gathers material at $p^0$.
With a slight abuse of notation we will use a superscript $^+$ ($^-$) to denote the next (previous) step $t^+$ ($t^-$) in which the agent gathers material at $p^0$ or terminates, e.g., $\emp{}^+$ denotes the set of empty nodes in step $t^+$.

\begin{lemma}
    \label{lem:enterCoat}
    If $P2$ holds in step $t$ and phase \coat{} is entered at some $\lp_i$ between step $t$ and $t^+$, then $\lp_{i-1}$ is the last node $v$ on $\lp$ with $v \in \links \cup \{\start{}\}$.
\end{lemma}

\begin{proof}
    As long as $p \in \links \cup \{\start{}\}$, the agent moves \lhr{} in phase \fetch{} which implies $\lp_{i-1} \in \links \cup \{\start{}\}$.
    It also moves \lhr{} if a successor of $p$ in $\seg(p)$ is a link.
    Hence, no successor $v$ of $\lp_{i-1}$ in $\seg(\lp_{i-1})$ is a link.
    Assume by contradiction that $\lp_{i-1}$ is not the last node on $\lp$ that is contained in $\links \cup \{\start{}\}$.
    Let $\lp_k$ be the first link after $\lp_{i-1}$, i.e., $\lp_k \in \links$ and $k > i-1$.
    Since no successor of $\lp_{i-1}$ in $\seg(\lp_{i-1})$ is a link, $\lp_k$ cannot be contained in $\seg(\lp_{i-1})$.
    Let $S = (v_0, ..., v_m)$ be the sequence of nodes with $v_0 = \lp_k$, $v_m = \start$ and $v_j = \pred(v_{j-1})$ for any $0 < j \leq m$.
    The agent's anchor changes only if there is no further successor in its current segment, i.e., \emph{after} moving $\lhr{}$ at some node $v$ of $\lp$ with $\suc(v) = v$.
    This implies $\suc(\pred(v)) = \pred(v)$ for all $v \in \lp$.
    It follows that $S$ must contain some node $v_j$ with $0 < j < m$ for which $v_j = \suc(\lp_i)$.
    Since $\suc(\lp_i) \notin \links$, there must exists some $v_{j'}$ with $j' < j$ for which $\pred(v_{j'}) \notin \links \cup \{\start{}\}$ which contradicts $P2$ and concludes the lemma.
    \end{proof}


\begin{lemma}
    \label{lem:cycles}
    If $\emp{}$ is connected and $P1$--$P4$ hold in step $t$, then $\emp^+$ is connected.
\end{lemma}

\begin{proof}
    Tiling a node $v \notin \links$ cannot disconnect $\emp{}$ by \cref{def:1link}.
    This covers the tile placement at the start of phase \coat{}, especially the case $\noCheck = true$, and the last placement before termination in phase \fetch{}.
    Thus, we must only consider cases in which a tile is placed at some link $v \in \links$.
    By \cref{lem:enterCoat}, the agent enters phase \coat{} at $\lp_i$ such that $\lp_{i-1}$ is the last node $w \in \lp$ with $w \in \links \cup \{\start{}\}$.
    Together with $P1$ follows that whenever it detects some node $w \in R_3(\lp_i,\lp_{i-1})$ with $w \in \links \cup \{\start{}\}$, then $w$ was visited in the previous execution of phase \fetch{}.
    Let $\lp_j$ be the last node of $\lp$ that is contained in $R_3(\lp_i,\lp_{i-1})$ with $j < i-1$, and $P$ be the shortest path from $\lp_i$ to $\lp_j$ in $\tri(R_3(\lp_i,\lp_{i-1}))$.
    Then $C = P \circ (\lp_{j+1}, \lp_{j+2},...,\lp_{i-2}, \lp_{i-1})$ is a simple cycle in $\tri$, where $\circ$ is the concatenation of paths.
    Placing a tile at $\lp_{i-1}$ or $\lp_{i-2}$ cannot disconnect $C$ since it is a cycle, and it cannot disconnect $\emp{}$ since $C$ contains nodes of all connected components of $\be(\lp_{i-1})$ (and $\be(\lp_{i-2})$, if $\lp_{i-2} \in \links$) by $P4$.
    \end{proof}

\begin{lemma}
    \label{lem:generatorSharedTiled} 
    If $\bo(v) \cap \bo(w) \neq \emptyset$, then tiling $v$ cannot increase $||\be(w)||$.
\end{lemma}

\begin{proof}
    The lemma follows trivially if $v \notin B(w)$.
    Consider arbitrary $v,w \in L$ with $v \in B(w)$.
    Since $\tri$ is a triangulation, the edge $\{v,w\}$ is contained in precisely two triangular faces, each of which contains another node $u_1$ and $u_2$, respectively.
    Since $B(v)$ and $B(w)$ are chordless, these are the only nodes adjacent to both $v$ and $w$, which implies that $B(v) \cap B(w) = \{u_1,u_2\}$ and that $\{u_1,v,u_2\}$ is connected in $B(w)$.
    If any $u_i$ is tiled, then $||\bo(w) \cup \{v\}||\leq ||\bo(w)||$.
    Thereby, $||\be(w) \setminus \{v\}|| = ||\bo(w) \cup \{v\}|| \leq ||\bo(w)|| = ||\be(w)||$.
    \end{proof}

We can now deduce the precise neighborhood of $v$ for the case where $v$ is both a link and a generator, i.e., $v \in \links \cap \gen$.
By \cref{def:coatability}, $B(v)$ contains at most six nodes, at least two of which must be tiled, as otherwise $v$ cannot be a link.
Hence, at most four empty nodes in at least two connected components of $\be(v)$ remain.
If each connected component has size at most two, then all nodes in $\be(v)$ share a tiled neighbor with $v$, which contradicts that $v$ is a generator by the previous lemma.
Hence, as a corollary we obtain the following:

\begin{corollary}
    \label{cor:generatorSize}
    For any $v \in \links \cap \gen$: $\bo(v)$ contains two connected components of size one, and $\be(v)$ contains two connected components of size one and three.
\end{corollary}

\begin{lemma}
    \label{lem:invariant}
    If $P1$--$P5$ hold and a node is tiled in step $t$, then either $P1$--$P5$ hold in step $t^+$ or the agent sets $\noCheck = true$ and $P1$--$P5$ hold in step $t^{++}$.
\end{lemma}

The proof of \cref{lem:invariant} is deferred to \cref{sec:appendix}.
Essentially, we distinguish the type of node $v$ that is tiled in step $t$, i.e., whether $v$ is not a link, a link but no generator, or a link and a generator, and finally whether $v + \rhr{} \in \links \cup \{\start\}$ (see line 9).
In all but the last case, we can show that $P1$--$P5$ immediately hold in step $t^+$, and in the last case, we know by the algorithm that $\noCheck{}$ is set to true after tiling $v$ in step $t^+$.
Here we can show that only property $P2$ is violated in step $t^+$, and only within the neighborhood of the previously placed tile.
Using \cref{cor:generatorSize}, we can precisely determine at which node the next tile is placed (with $\noCheck{} = true$), and that this tile placement restores $P2$ in step $t^{++}$.
In the other cases the properties mostly follow from \cref{lem:generatorSharedTiled}.

All properties hold in an initial configuration, and they are maintained as invariants by \cref{lem:invariant}.
By \cref{lem:enterCoat}, the agent eventually terminates in some step $t^*$, and by \cref{lem:cycles} $\emp{}$~never disconnects.
Hence, $\occ{}^{t^*} = L$ holds after termination which concludes the following:

\begin{theorem}
    \label{thm:algorithm}
    Following \cref{alg:algorithm}, a finite-state agent solves the coating problem on $\tri$, given a configuration $C^0 = (\occ^0,\obj{},p^0)$ that is coatable w.r.t. $\tri$.
\end{theorem}

\subsection{Runtime Analysis}
\label{subsec:optimality}

Since the agent does not sense any node outside of $N_1(p)$ in its look-phase, exploring $R_i(p,q)$ requires additional steps.
From $R_i(p,q) \subset N_{i+2}(p)$ it follows that the number of steps is upper bounded by $2 \cdot |N_{i+2}(p)| = \O(|N(p)|^i)$. 
Since $i$ is a constant and $\tri$ has constant degree, each execution of phase \coat{} takes $\O(1)$ steps.
Each execution of phase \fetch{} takes $\O(|\lp|)$ steps as the agent traverses a sub-path of $\lp$ twice.
Since $\lp$ is simple, it follows that $|\lp| = \O(n)$, where $n = |L|$.
The agent can place at most $n$ tiles until $L = \occ{}$, thereby performs at most $n$ executions of \coat{} and \fetch{}, which results in $\O(n^2)$ steps in total.

An agent $\tilde{r}$ with unlimited memory and global vision can reach any node via a shortest path instead of sticking to the boundary of tiled nodes.
Except for the last placed tile it must always return to the material depot which implies that the last tile is placed at a node $w$ with maximum distance to $p^0$.
In the worst case, the surface graph is the triangulation of an object resembling a straight line such that $d_L(p^0,w) = \Theta(n)$. 
Each node on the shortest path $P$ from $p^0$ to $w$ must be tiled.
Hence, $\tilde{r}$ takes at least
$2\left( \sum_{u \in P} d_L(p^0,u)\right) - d_L(p^0,w) = \left(\sum_{i=1}^{\Theta(n)} 2 \cdot i\right) - \Theta(n) = \Theta(n^2)$ steps which implies worst-case optimality of \cref{alg:algorithm}.

\section{Coating in the 3D Hybrid Model}
\label{sec:construction}

In this section, we apply our coating algorithm to the 3D hybrid model.
We first define a triangulation on nodes of $L$ with degree $\Delta \leq 8$, and afterwards construct a virtual graph on which we emulate our algorithm using $2^{2\Delta}$ types of tiles.

\begin{figure}[t]
    \centering
    \begin{minipage}{.43\textwidth}
        \centering
        \includegraphics[width=\linewidth]{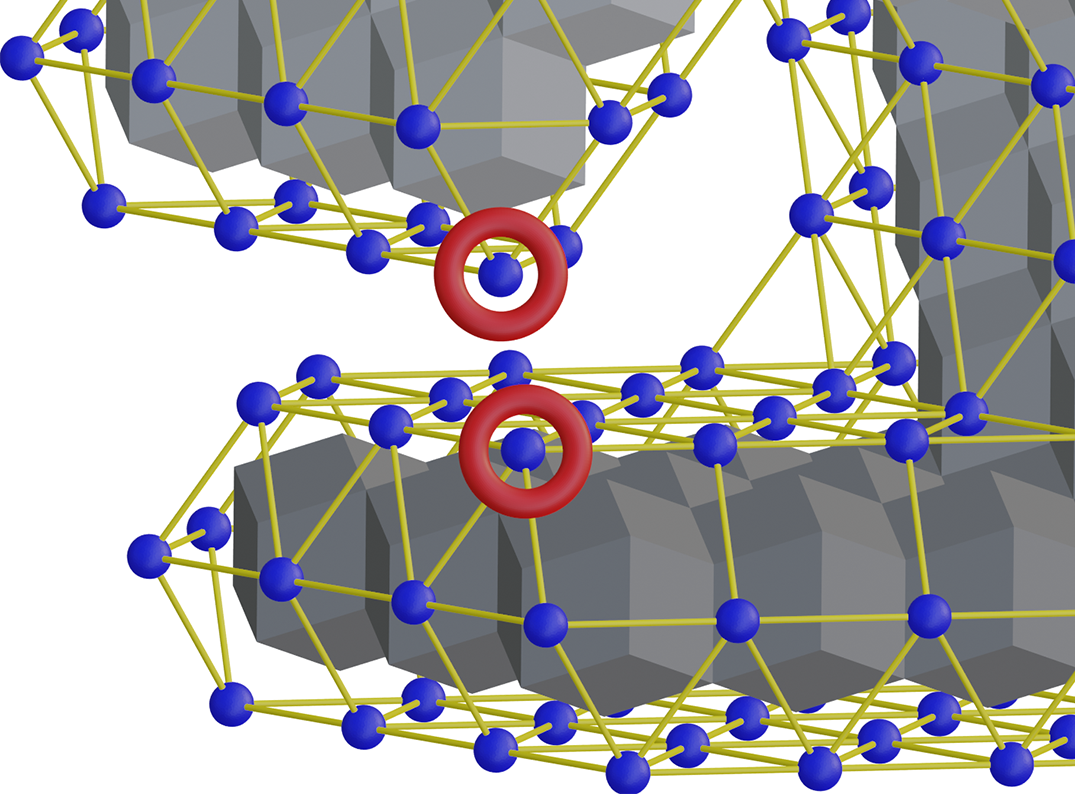}
        \caption{A snapshot of $\surf$: the circled nodes are adjacent in $G(L)$.}
        \label{fig:coatingLayerGraph}
    \end{minipage}%
    \hfill
    \begin{minipage}{.55\textwidth}
        \centering
        \includegraphics[width=\linewidth]{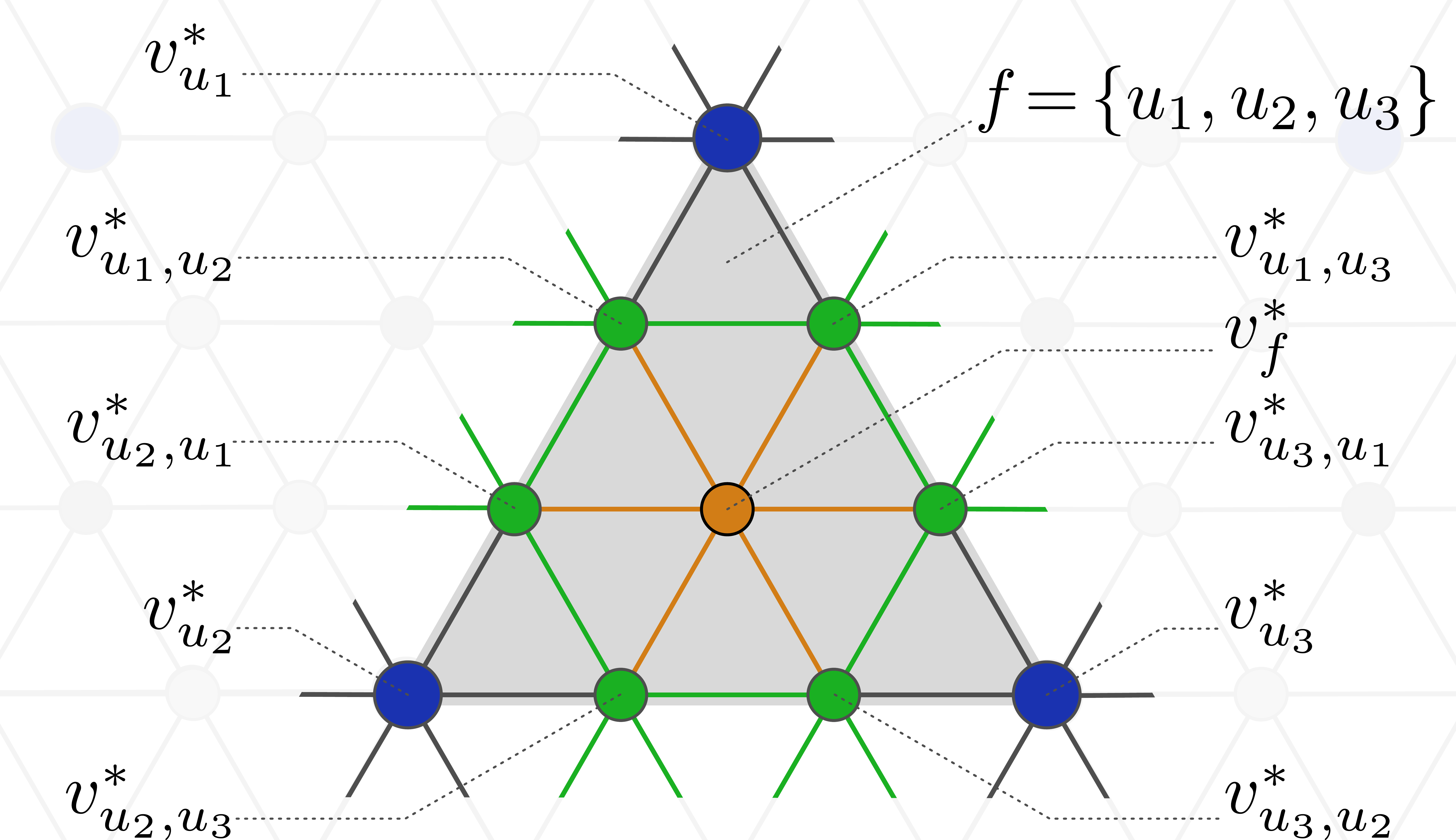}
        \caption{A triangular face $f$ of $\tri$ and its corresponding virtual edges and nodes in $\tri^*$.}
        \label{fig:emulation}
    \end{minipage}
\end{figure}

\begin{figure}[b]
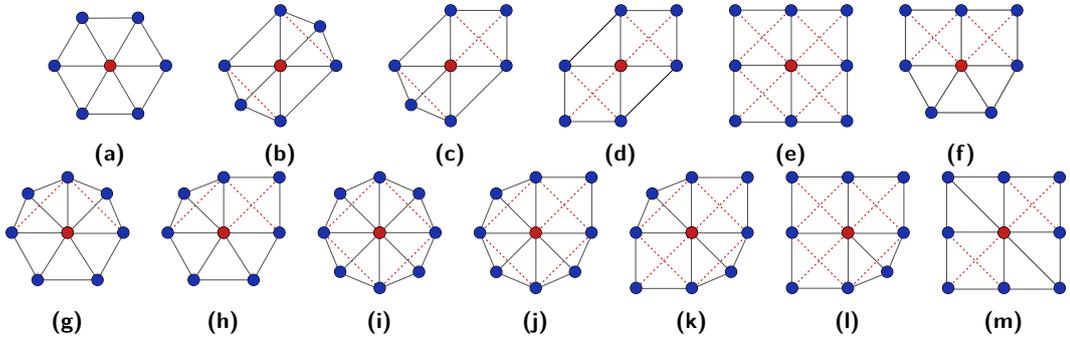

    \centering
    \hfill
    \foreach \x in {a,...,f}{%
        \begin{subfigure}[c]{0.12\linewidth}
            \includegraphics[width=\linewidth]{boundaries_\x}
            \subcaption{}
            \label{subfig:triangulation_\x}
        \end{subfigure}%
        \hfill
    }%
    \null\hfill
    \\
    \foreach \x in {g,...,m}{%
        \begin{subfigure}[c]{0.12\linewidth}
            \includegraphics[width=\linewidth]{boundaries_\x}
            \subcaption{}
            \label{subfig:triangulation_\x}
        \end{subfigure}%
        \hfill
    }%
    \caption{All possible arrangements of faces in $\surf$ apart from rotation. Dashed edges indicate that the distance between its endpoints is precisely two w.r.t. $G$.}
    \label{fig:triangulation}
\end{figure}


Recall the definition of graph $G = (V,E)$ and its embedding in $\mathbb{R}^3$ from \cref{sec:model}.
Define $\surf = (L,E')$ as the subgraph of $G(L)$ that contains only those edges $\{v,w\}$ for which $v$ and $w$ share adjacent object neighbors (see \cref{fig:coatingLayerGraph}), i.e., $E' = \{\{v,w\}\mid d_\obj{}(N_1(v), N_1(w)) \leq 1\}$.
We can view $\surf$ as embedded on the surface of our 3D object.
That embedding contains triangular and tetragonal faces (see \cref{fig:coatingLayerGraph}) where tetragonal faces can occur in one of three orientations:
(1) $v, v+ \NE, v + \NE + \USE, v + \USE$, (2) $v, v + \NW, v + \NW + \UNE, v + \UNE$, and (3)~$v, v + \N, v + \N + \UW, v+ \UW$.
Apart from rotation, \cref{fig:triangulation} shows all possible arrangements of faces within $\surf$. 
We define the class of \emph{smooth objects} $\mathcal{S}$ as all objects for which $\surf$ contains only the cases (a)--(f) from \cref{fig:triangulation}.
Let $\tri$ be the triangulation of $\surf$ in which the same diagonal edge is added for each tetragonal face of the same orientation (1)--(3) (since we want the agent to be able to deduce the triangulation).
Since $d_L(v,w) = 1$ w.r.t. $\tri$ implies $d_L(v,w) \leq 2$ w.r.t. $\surf$, the agent can emulate moving on $\tri$ with a multiplicative time and memory overhead of at most two.
It is easy to see that $B(v)$ is chordless and $v$ has degree at most six for all $v \in L$ within the class $\mathcal{S}$.
Together with \cref{thm:algorithm} follows:

\begin{theorem}
    \label{thm:coatableSingleType}
    A finite-state agent with a single tile type solves the coating problem on any object $\obj \in \mathcal{S}$ with coating layer $L$ in $\O(n^2)$ steps, where $n = |L|$.
\end{theorem}

\subsection{Emulation of Coatable Surface Graphs}
\label{subsec:emulation}

Consider an arbitrary triangulation $\tri = (L,E)$ of constant degree $\Delta$ and an initially valid configuration $C_0$.
We construct a virtual graph $\tri^* = (L^*,E^*)$ with virtual initial configuration $C^{0*}$ such that $\tri^*$ is coatable w.r.t. $C^{0*}$.
During that construction, we define a partial surjective function $\mathcal{R}: L^* \rightarrow L$ that maps virtual nodes to real nodes.
We show that an agent $r$ operating on $\tri$ w.r.t. $C_0$ with $2^{2\Delta}$ tile types can emulate an agent $r^*$ that executes \cref{alg:algorithm} on $\tri^*$ w.r.t. $C^{0*}$ such that throughout the emulation $\mathcal{R}(p^*) = p$.

\subsubsection{Virtual Graph Construction}

The virtual graph $\tri^*$ is the result of subdividing each face of $\tri$ into nine triangular faces (see \cref{fig:emulation}).
The node set $L^*$ contains a virtual node $v^*_u$ for each node $u \in L$, two virtual nodes $v^*_{u,w}$ and $v^*_{w,u}$ for each edge $\{u,w\} \in E$, and a virtual node $v^*_f$ for each triangular face $f$ of $\tri$.
For each edge $\{u,w\} \in E$ the edge set $E^*$ contains three virtual edges $\{v^*_u,v^*_{u,w}\}$, $\{v^*_{u,w}, v^*_{w,u}\}$ and $\{v^*_{w,u}, v^*_w\}$.
For each triangular face $f = \{u_1,u_2,u_3\}$ of $\tri$, $E^*$ contains six virtual edges $\{v^*_f,v^*_{u_i,u_j}\}$ and three virtual edges $\{v^*_{u_i,u_j}, v^*_{u_i,u_k}\}$, where $u_i, u_j, u_k \in f$ are pairwise distinct.
We define $\mathcal{R}(v^*_{u,w}) = u$ for any virtual node $v^*_{u,w} \in L^*$. 
Consider an arbitrary but fixed order on the vectors $\overrightarrow{\X_1},...,\overrightarrow{\X_m}$ that correspond to edges in the embedding of $\tri$.
Let $\pi$ represent that order, i.e., $\pi(\overrightarrow{\X_i}) = i$.
For some face $f = \{u_1,u_2,u_3\}$ of $\tri$, we define $\mathcal{R}(v^*_f) = u_i$, where $u_i$ is the node minimizing $\pi(\overrightarrow{u_i} - \overrightarrow{u_j})$ for any $u_i,u_j \in f$ with $i \neq j$.
We define the virtual initial configuration $C^{0*}$ such that all $v^*_u$ are tiled, i.e., $\occ^{0*} = \cup_{u \in L} v^*_u$, $p^{0*} = v^*_{p^0}$ and assume a material depot of size at least $|L^*|-|L|$ at $v^*_{p^0}$.

\begin{lemma}
    \label{lem:virtualGraph}
    $C^{0*}$ is coatable w.r.t. $\tri^*$.
\end{lemma}

\begin{proof}
    Each face of $\tri$ is triangular, and two virtual nodes are added for each edge of $\tri$.
    Hence, $|B(v^*_f)|= 6$ w.r.t. $\tri^*$ for any face $f$ of $\tri$.
    Any $v^*_{u,w}$ is adjacent to $v^*_{f_1}$ and $v^*_{f_2}$, where $f_1,f_2$ are the two faces of $\tri$ that both contain $u$ and $w$, to two nodes $v^*_{u,w_1}, v^*_{u,w_2}$, where $w_1 \in f_1$ and $w_2 \in f_2$, and to $v^*_{w,u}$ and $v^*_u$.
    Hence, $|B(v^*_{u,w})| = 6$ w.r.t. $\tri^*$ for any edge $\{u,w\}$ of $\tri$.
    Any other virtual node is initially tiled, which implies $|B(v^*)| \leq 6$ for any $v^* \in \emp^*$.
    By construction, each initially tiled node is isolated, i.e., $d(v^*,w^*) \geq 3$ for any $v^*,w^* \in \occ{}^{0*}$.
    Since $\tri^*$ is connected, it follows that $\emp^{0*}$ is connected and $\be(v^*)$ is connected for any $v^* \in \emp^{0*}$, i.e., $\links^{0*} = \emptyset$.
    Hence, each property of \cref{def:coatability} is satisfied.
    \end{proof}

\begin{lemma}
    \label{lem:emulation}
    A finite-state agent can emulate \cref{alg:algorithm} on $\tri^*$ in $\O(\Delta^2n^2)$ steps while moving and placing tiles of at most $2^{2\Delta}$ types on $\tri$.
\end{lemma}

\begin{proof}
    Let $F^* \subset L^*$ be the set of virtual nodes $v^*_f$ that correspond to some face $f$ of $\tri$ in the construction of $\tri^*$.
    Since $\tri^*(L^* \setminus F^*)$ is a subdivision of $\tri$, it can be embedded in the same 3D surface as $\tri$ using vectors that are collinear to vectors in the embedding of $\tri$.
    It follows that we can use the same fixed order $\pi$ from the construction of $\tri^*$.

    In the following, we define for each node $u \in L$ a bit-sequence $x(u) = (x_1,...,x_{2\Delta})$ that encodes the occupation of all nodes $v^* \in L^*$ with $\mathcal{R}(v^*) = u$, where a $0$ encodes an empty, and a $1$ encodes an occupied virtual node.
    By the construction of $\tri^*$, there are at most $2 \Delta$ nodes $v^*$ with $\mathcal{R}(v^*) = u$ such that $2 \Delta$ bits suffice.
    The order of bits in $x(u)$ is uniquely given by $\pi$ where the first $\Delta$ bits encode virtual nodes that correspond to edges of $\tri$, and the following bits encode virtual nodes that correspond to faces of $\tri$.
    There is no bit for the virtual node $v^*_u \in L^*$ since it is initially occupied and remains occupied until termination by following \cref{alg:algorithm}.
    In fact, $\mathcal{R}$ is undefined for $v^*_u \in L^*$.

    Consider an agent $r$ on $\tri$ that utilizes $k = 2^{2\Delta}$ types of passive tiles.
    Each tile type uniquely describes a bit-sequence of length $\log(k) = 2\Delta$ such that $r$ emulates an agent $r^*$ on $\tri^*$ with initial configuration $C^{0*}$ as follows:
    If $r^*$ moves from $v^*$ to $w^*$, then $r$ moves from $\mathcal{R}(v^*)$ to $\mathcal{R}(w^*)$ (if $\mathcal{R}(v^*) \neq \mathcal{R}(w^*)$).
    If $r^*$ places a tile at $v^*$ and $\mathcal{R}(v^*)$ is empty, then $r$ places a tile at $\mathcal{R}(v^*)$ that corresponds to the bit-sequence $x$ in which only $v^*$ is encoded as occupied, otherwise $r$ incorporates the occupation of $v^*$ by changing the tile type.
    If $r^*$ gathers material and $r$ carries no material, then $r$ also gathers material.

    By \cref{thm:algorithm} and \cref{lem:virtualGraph}, $r^*$ solves the coating problem on $\tri^*$.
    Since $\mathcal{R}$ is surjective and any node $\mathcal{R}(v^*) \in L$ is occupied, if $v^*\in L^*$ is occupied, the emulation solves the coating problem on $\tri$ in $\O(|L^*|^2) = \O(\Delta n)$ steps.
    \end{proof}

Our final theorem follows from the virtual graph construction on top of our triangulation $\tri$ of $\surf$ (with $\Delta \leq 8$) and the previous lemma:

\begin{theorem}
    \label{thm:coatableManyType}
    A finite-state agent utilizing constantly many tile types can solve the coating problem on arbitrary objects in worst-case optimal $\O(n^2)$ steps.
\end{theorem}
\section{Future Work}
\label{sec:conclusion}

We provided an algorithm that solves the coating problem in the 3D hybrid model in worst-case optimal $\O(n^2)$ steps given that the initial configuration w.r.t. the surface graph fulfills the property of coatability as specified in \cref{def:coatability}.
While the algorithm solves the problem directly in the class of smooth objects, there are certainly surface graphs that violate coatability.
We bypassed this problem by emulating our algorithm on a subdivision of these surface graphs using $2^{2\Delta}$ types of tiles.
A natural question for future work is whether solving the problem with a single tile type is in fact impossible, and if so, then what is the lowest number of tile types required to solve it.
Another open question is how far our worst-case optimal solution is off from the best case solution.

%
%
%
\bibliographystyle{plainurl}
\bibliography{bibliography}
\appendix
\section{Deferred Proofs}
\label{sec:appendix}

\setcounter{lemma}{6}
\begin{lemma}
    \label{lem:invariantAppendix}
    If $P1$--$P5$ hold and a node is tiled in in step $t$, then either $P1$--$P5$ hold in step $t^+$ or the agent sets $\noCheck = true$ and $P1$--$P5$ hold in step $t^{++}$.
\end{lemma}

\begin{proof}
    \setcounter{lemma}{0}
    The lemma is proven by a case distinction on the type of node $v$ that is tiled in step $t$.
    For the sake of clarity, each case is proven individually in the following claims. 
    \begin{claim} 
        If $P1$--$P5$ hold in step $t$, and a tile is placed at some $\lp_i \notin \links \cup \{\start{}\}$ with $\noCheck = false$, then $P1$--$P5$ hold in step $t^{+}$. 
    \end{claim}
    \begin{claimproof}
        We first show that $P4$ holds in step $t^+$.
        The tile is not placed at a link, which by \cref{lem:enterCoat} implies that $\lp_i$ is the node at which the agent enters phase \coat{} and $\lp_{i-1}$ is the last node on $\lp$ that is contained in $\links \cup \{\start{}\}$.
        Since $\noCheck = false$ by assumption, $R_3(\lp_i,\lp_{i-1})$ must have been searched by the agent and cannot contain any link.
        Placing a tile at $\lp_i$ can only generate links in $\be(\lp_i)$, and the only possible node $w \in \be(\lp_i)$ that is already a link, i.e., $||\be(w)|| > 1$, is $\lp_{i-1}$.
        Since $\lp_{i}$ shares a tiled neighbor with $\lp_{i-1}$, $\lp_{i}$ cannot increase $||\be(\lp_{i-1})||$ by \cref{lem:generatorSharedTiled} such that $P4$ holds in step $t^+$.

        To show the remaining properties, we distinguish two cases based on the number of empty neighbors of $\lp_i$.
        First, consider the case $|\be(\lp_i)| > 1$.
        Let $\tilde{N}_3(\lp_i)$ be the $3$-neighborhood of $\lp_i$ w.r.t. $\tri(\emp{} \setminus \{\lp_{i-1}\})$.
        By \cref{def:range}, $R_3(\lp_i,\lp_{i-1})$ contains $\pred(v)$ for all $v \in \lp \cap \tilde{N}_3(\lp_i)$.
        $R_3(\lp_i,\lp_{i-1})$ does not contain $\start$ or any link, which together with $P2$ implies that $j > i - 1$ for any $\lp_j \in \tilde{N}_3(\lp_i) \cap \lp$.
        Thereby, the subpath of $\lp$ from $\start$ to $\lp_{i-1}$ does not change from step $t$ to $t^+$, and due to the fact that a tile is placed at $\lp_i$, each node in $\be(\lp_i)$ is contained in $\lp^+$.
        Together with $\tilde{N}_1(v) \subseteq \tilde{N}_3(\lp_i)$ for all $v \in \be(\lp_i)$, it follows that $P3$ holds in step $t^+$.
        Since each link that is generated in step $t$ is contained in $\be(\lp_i)$, $P1$ holds in step $t^+$ as well.
        It remains to show $P2$ and $P5$.
        $\lp_i$ cannot consume $\lp_{i-1}$, i.e., $\lp_{i-1} \in \links^+ \cup \{\start{}\}$, as otherwise $\lp_{i-1}$ is contained in a connected component of $\be(\lp_i)$ of size one, which would contradict $\lp_i \notin \links$ or $|\be(\lp_i)| > 1$.
        $\lp_i$ cannot generate $\lp_{i+1}$ by \cref{lem:generatorSharedTiled}, i.e., $\lp_{i+1} \notin \links^+$.
        $\lp_{i-1}$ is the only node in $\be(\lp_i)$ without a predecessor and $\lp_{i+1}$ is the only node in $\be(\lp_i)$ without a successor.
        Then for all $v \in \be(\lp_i)$ with $v \neq \lp_{i-1}$ holds that $\pred^+(v) = \lp_{i-1} \in \links^+ \cup \{\start{}\}$  and $\suc^+(v) = \lp_{i+1} \notin \links^+$, i.e., $P2$ and $P5$ hold in step $t^+$.

        Second, consider the case $|\be(\lp_i)| = 1$.
        By symmetry, $\lp_i$ is a node of a connected component of $\be(\lp_{i-1})$ of size one.
        By $P4$, $\lp_i$ consumes $\lp_{i-1}$ (if $\lp_{i-1} \in \links$) and it cannot generate any link since it has no other empty neighbor.
        Then $P1$--$P4$ hold trivially in step $t^+$.
        If $\lp_i$ is a successor of $\lp_{i-1}$ in $\seg(\lp_{i-1})$ in step $t$, then $\lp_{i-1}$ has no successor in step $t^+$, i.e., $\suc(\lp_{i-1})^+ = \lp_{i-1}$ and $P5$ holds.
    \end{claimproof}
    \begin{claim}  
        If $P1$--$P5$ hold in step $t$, and a tile is placed at some $\lp_i \in \links$ with $\lp_{i} \notin \gen$, then $P1$--$P5$ hold in step $t^{+}$.
    \end{claim}
    \begin{claimproof}
        There are two cases: (1) the agent enters \coat{} at $\lp_{i+2}$ and moves $\rhr{}$ twice, i.e., $\lp_{i+1} \in \links \cap \gen$, or (2) it enters \coat{} at $\lp_{i+1}$.
        In case (1), $\lp_i$ consumes $\lp_{i+1}$ by \cref{cor:generatorSize}, as $\lp_{i}$ must be contained in a connected component of $\be(\lp_{i+1})$ of size one.
        In both cases $\lp_i \notin \lp^+$ since $\lp_i$ is tiled in step $t^+$.
        Together with \cref{lem:enterCoat}, it follows that only the last link of $\lp$ (and second last link in case (1)) is consumed and no link is generated, which implies that the sub-path of $\lp$ from $\start$ to $\lp_{i-1}$ is identical in step $t$ to $t^+$ $P1$--$P4$ hold.

        If the connected component $K$ of $\tri(\be(\lp_{i}))$ that contains $\lp_{i-1}$ has size one, then $\lp_i$ consumes $\lp_{i-1}$ by $P4$ such that $\suc^+(\lp_{i-1}) = \lp_{i-1}$ and $P5$ holds.
        Otherwise, for any $\lp_j \in K$ with $j \neq i-1$ it holds that $j > i$ by $P3$, and thereby $\lp_j \notin \links$ by $P1$.
        Together with the lemma's assumption $\lp_i \notin \gen$ follows that $\lp_j \notin \links^+$, especially $\suc(\lp_i^+) \notin \links$.
        Hence, $P5$ holds in step $t^+$.
    \end{claimproof}
    \begin{claim} 
        \label{claim:firstFig}
        If $P1$--$P5$ hold in step $t$, and a tile is placed at some $\lp_i \in \links \cap \gen$ with $\lp_{i-1} \notin \links \cup \{\start{}\}$, then $P1$--$P5$ hold in step $t^{++}$.
    \end{claim}
    \begin{claimproof}
        First, we deduce the neighborhood of $\anchor(\lp_i), \lp_i$ and $\lp_{i+1}$.
        For ease of reference, refer to \cref{fig:noCheck}.
        Since the agent places a tile at some $\lp_i \in \links \cap \gen$, $\noCheck{}$ must be $false$ in that execution of phase \coat{} and the agent moves \rhr{} twice before it places a tile and sets $\noCheck$ to $true$.
        By \cref{lem:enterCoat}, phase \coat{} is entered at $\lp_{i+2} \notin \links$ with $\lp_{i+1}, \lp_{i} \in \links \cap \gen$, and $\lp_{i+1}$ is the last $v \in \lp$ with $v \in \links \cup \{\start{}\}$.
        By \cref{cor:generatorSize}, both $\be(\lp_{i+1})$ and $\be(\lp_{i})$ contain two connected components of size one and three.
        $\lp_{i+2}$ must be contained in a connected component of $\be(\lp_{i+1})$ of size three, as otherwise $\lp_{i+2} \in \links$ which contradicts that the agent enters phase \coat{} at $\lp_{i+2}$, or $\be(\lp_{i+2}) = \{\lp_{i+1}\}$, which implies $R_3(\lp_{i+2},\lp_{i+1}) = \emptyset$ and contradicts that a tile is placed at a link.
        Thereby, $\lp_{i+2}, \lp_{i+1}, \lp_{i}$ and $\lp_{i-1}$ are all contained in the same segment $\seg(\lp_i)$.
        The lemma's assumption $\lp_{i-1} \notin \links \cup \{\start{}\}$ implies that there must exist another node $\lp_{i-2}$ in that segment since otherwise $\lp_{i-1} = \pred(\lp_i) \notin \links \cup \{\start{}\}$ which would contradict $P2$.
        Next, we show that $\lp_i$ and $\lp_{i+1}$ were generated by $\anchor(\lp_i)$ in some step $t' < t$.
        Initially, $\links^0 = \emptyset$, which implies that the links at $\lp_i$ and $\lp_{i+1}$ were generated by one of the two tiled nodes in $B(\lp_i) \cap B(\lp_{i+1})$.
        Let $u$ be the node that generated $\lp_i$ and $\lp_{i+1}$ in step $t'$.
        By \cref{def:coatability}, it holds that $|B(u)|\leq 6$.
        By the above deduction of the neighborhood of $\lp_i$ and $\lp_{i+1}$, $\be(u)$ contains a connected component of size four.
        Hence, $B(u) \setminus \be(u)$ is connected and contains at most two nodes.
        Since the agent never disassembles any tile, this implies that $u$ was not a link in step $t'$.
        Node $u$ generates both $\lp_i$ and $\lp_{i+1}$, and $\lp_i$ is visited before $\lp_{i+1}$ in step $t$, which implies $\anchor(\lp_i) = u$ and thus $|B(\anchor(\lp_i))| \leq 6$.
        Hence, $\seg(\lp_i)$ contains five nodes $\sigma_j$ with $j \in \{i-2,...,i+2\}$ and $\anchor(\lp_i)$ has one tiled neighbor, which precisely results in the local configuration depicted by \cref{fig:noCheck} apart from rotation.

        \begin{figure}[!t]
            \centering
            \begin{subfigure}[c]{0.33\linewidth}
                \includegraphics[width=\linewidth]{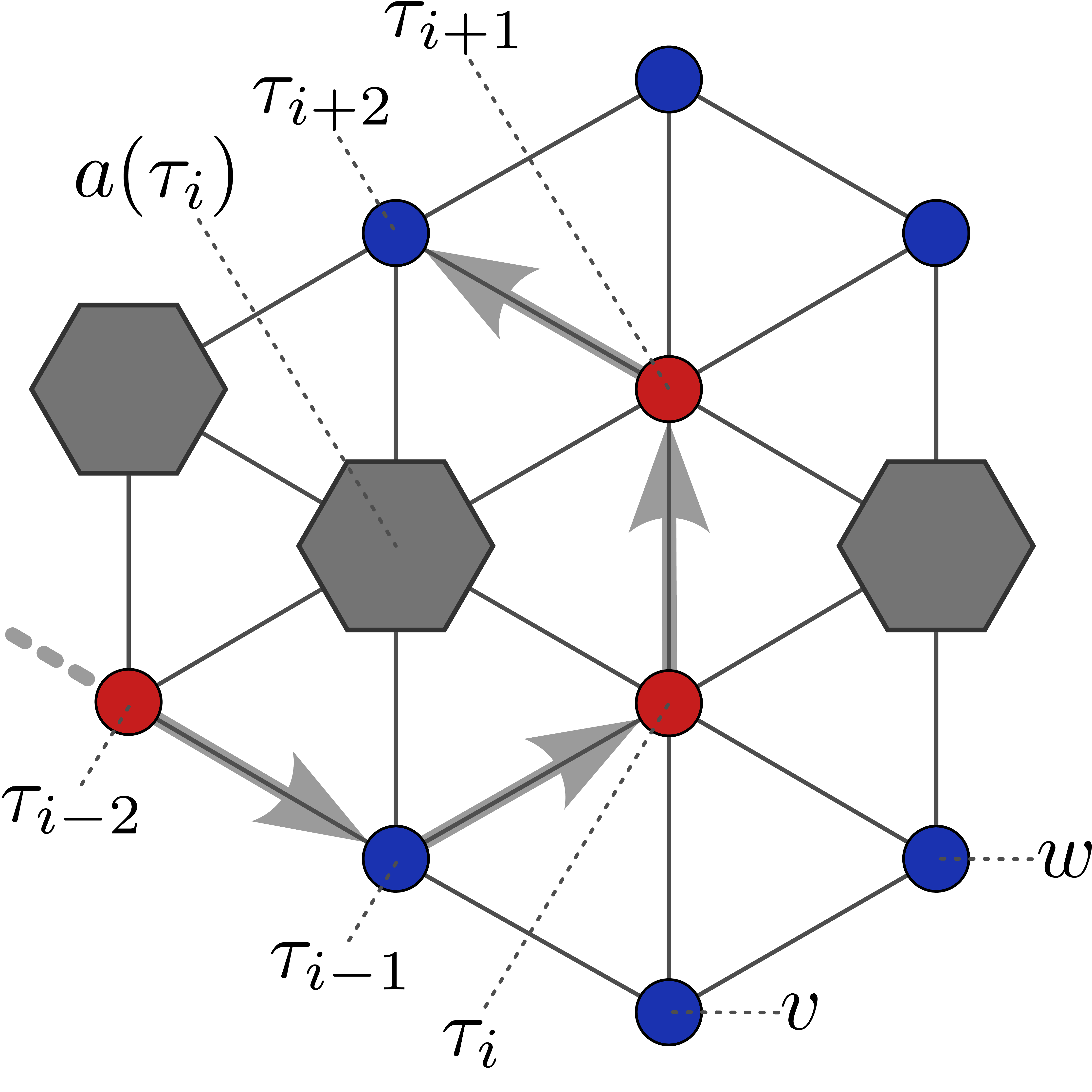}
                \subcaption{step $t$}
                \label{fig:noCheck1}
            \end{subfigure}%
            \begin{subfigure}[c]{0.33\linewidth}
                \includegraphics[width=\linewidth]{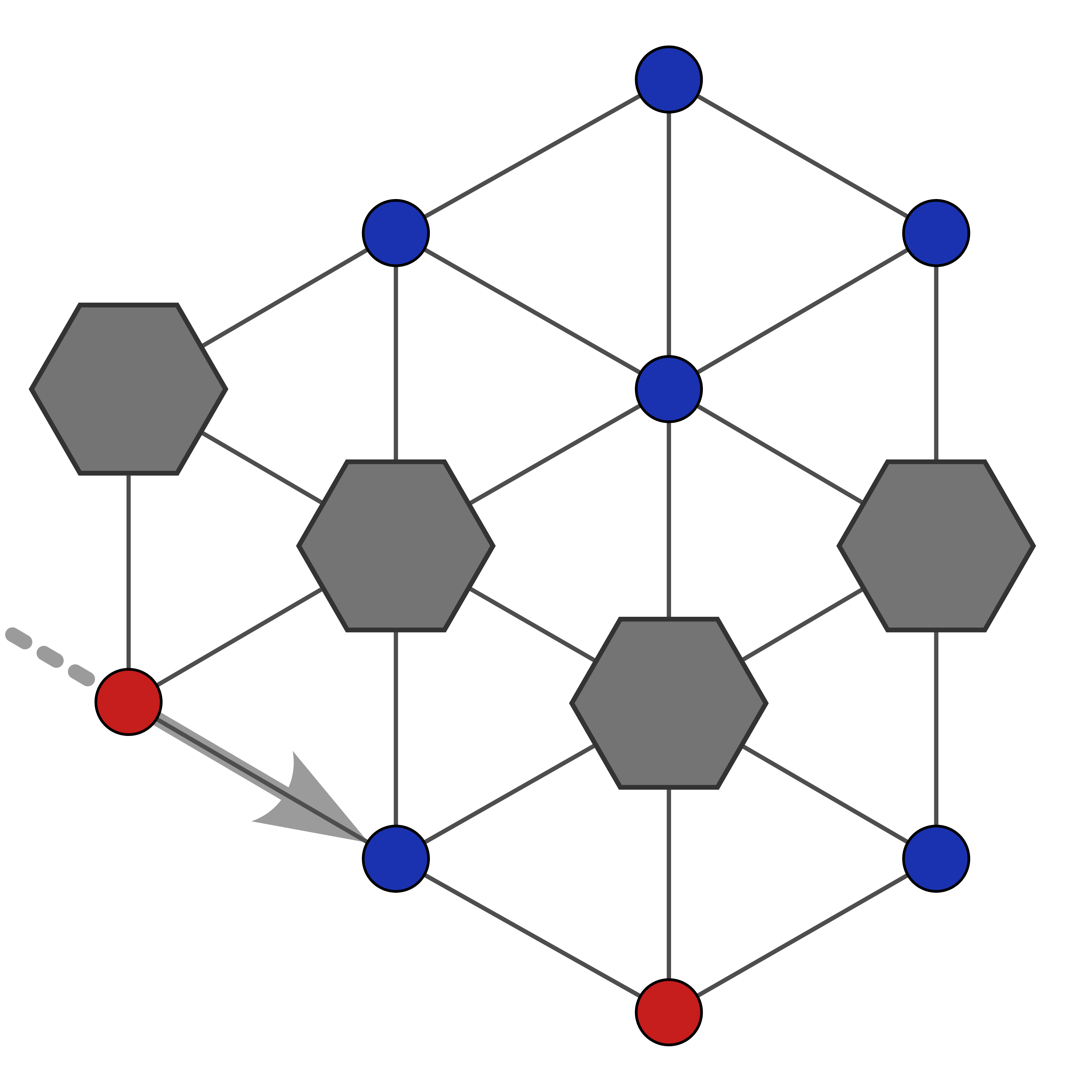}
                \subcaption{step $t^+$}
                \label{fig:noCheck2}
            \end{subfigure}%
            \begin{subfigure}[c]{0.33\linewidth}
                \includegraphics[width=\linewidth]{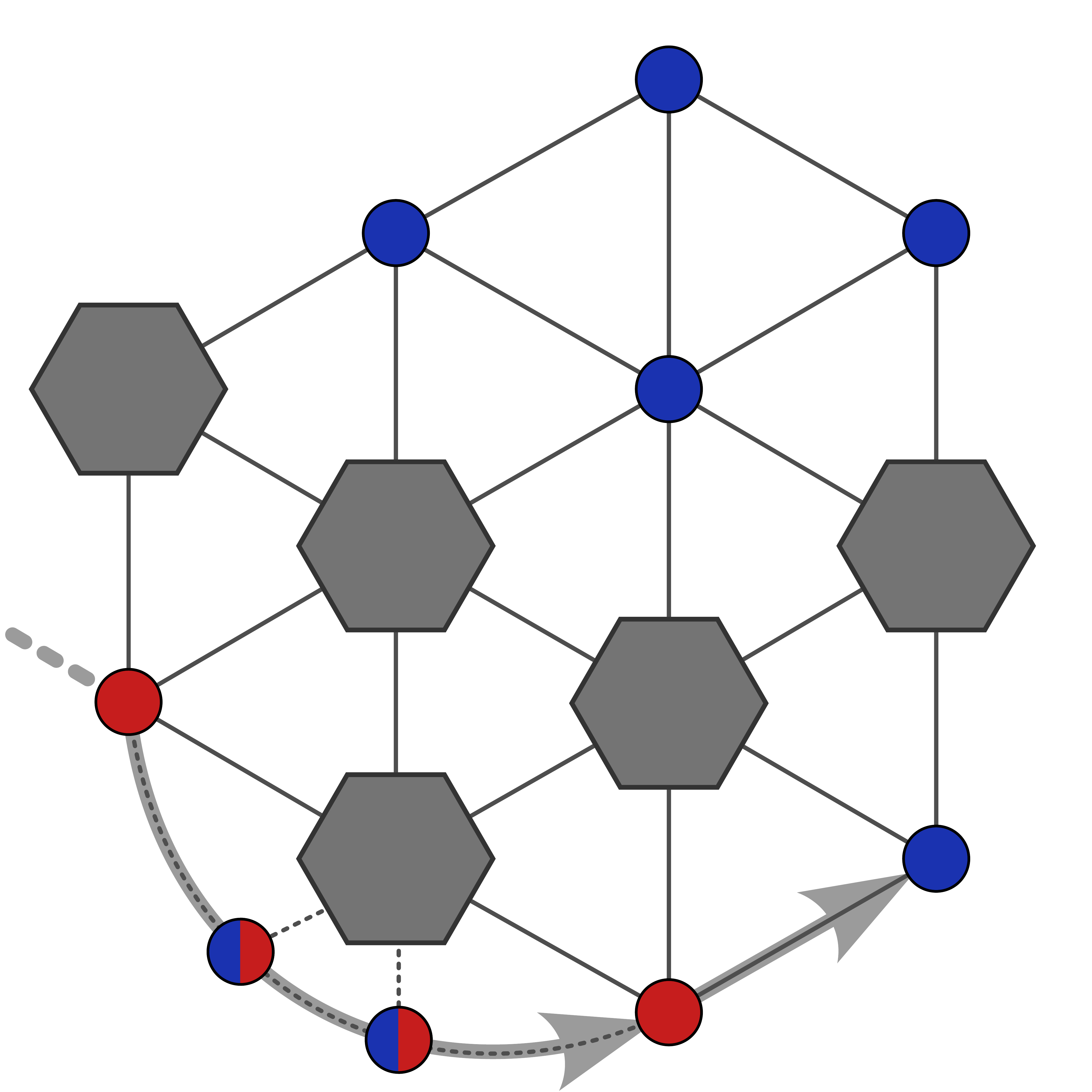}
                \subcaption{step $t^{++}$}
                \label{fig:noCheck3}
            \end{subfigure}%
            \caption{Local configuration in the proof of Claim 3 together with the traversed path in phase \fetch{} (tiled nodes are depicted as hexagons, empty nodes as circles, links are red circles).}
            \label{fig:noCheck}
        \end{figure}

        Second, we consider the situation after a tile is placed at $\lp_i$ between step $t$ and $t^+$ and show that $P2$ holds in step $t^{++}$.
        Let $v$ be the node that is generated by $\lp_{i}$, and $w$ be the other node in the connected component of $\be(\lp_i)$ distinct from $\lp_{i-1}, \lp_{i+1}$ and $v$.
        As can be seen in \cref{fig:noCheck1}, $\lp_i$ consumes $\lp_{i+1}$, it generates $v$ and it cannot generate $w$ or $\lp_{i-1}$ since they share a tiled neighbor with $\lp_i$.
        Hence, in step $t^+$ the neighborhoods are precisely depicted by \cref{fig:noCheck2}.
        Since $\pred^+(v) = \lp_{i-1}$ and $\lp_{i-1} \notin \links^+$, $P2$ is violated and the agent enters phase \coat{} with $\noCheck{}$ set to $true$ at $\lp_{i-1}$ between step $t^+$ and $t^{++}$ without visiting $v$.
        In this case it places a tile at $\lp_{i-1}$ without searching for links in $R_3(\lp_{i-1},\lp_{i-2})$.
        Afterwards, it holds that $\pred^{++}(v) = \lp_{i-2}$ such that $P2$ holds again in step $t^{++}$.
        For completeness, \cref{fig:noCheck3} shows the neighborhoods w.r.t. step $t^{++}$.

        Third, we show that $P1$ and $P3$--$P5$ are maintained.
        Recall that we already concluded that $\lp_{i}$ and $\lp_{i+1}$ were generated by $\anchor(\lp_i)$ in some prior step $t' < t$.
        Note that precisely two nodes in $B(\lp_{i-1})$ are tiled in step $t^+$ (see \cref{fig:noCheck2}).
        Since $B(\anchor(\lp_i))$ contains only one tiled node in step $t$, by contraposition it follows that in step $t'$ the tile at $\anchor(\lp_i)$ is placed with $\noCheck{} = false$.
        No node from $\be(\anchor(\lp_i))$ was ever tiled prior to step $t$, and $\anchor(\lp_i)$ cannot generate $\lp_{i+2}$ in step $t'$ by \cref{lem:generatorSharedTiled}.
        It follows that no tile was placed between step $t'$ and $t$, i.e., $t' = t^-$.
        This implies that if $\anchor(\lp_i)$ were empty in step $t$, then $R_3(\anchor(\lp_i), \lp_{i-2})$ contains neither $\start$ nor any link.
        The only links that are generated and not consumed between step $t^-$ and $t^{++}$ are contained in $\be(\lp_{i-1})$ (note that $v \in \be(\lp_{i-1})$ as well).
        Hence, $P1$ and $P3$--$P4$ hold in step $t^{++}$ analogous to Claim 1, and $P5$ holds since $w \notin \links^{++}$ and $\suc^{++}(w) = w$.
    \end{claimproof}
    
    \begin{figure}[!b]
        \centering
        \begin{minipage}[t][][b]{.63\textwidth}
            \begin{subfigure}[c]{0.5\linewidth}
                \includegraphics[width=\linewidth]{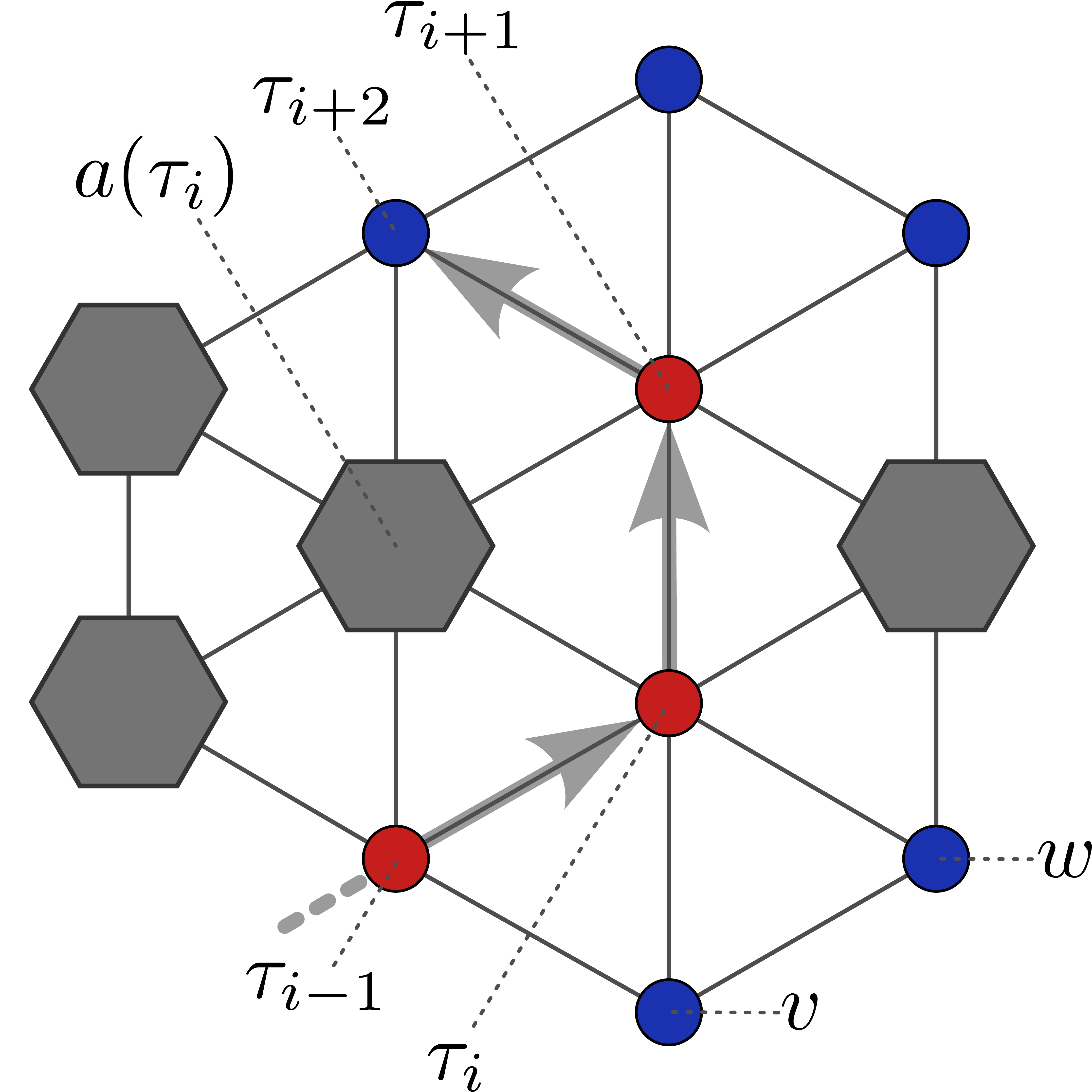}
                \subcaption{step $t$}
                \label{fig:invariantGen1}
            \end{subfigure}%
            \begin{subfigure}[c]{0.5\linewidth}
                \includegraphics[width=\linewidth]{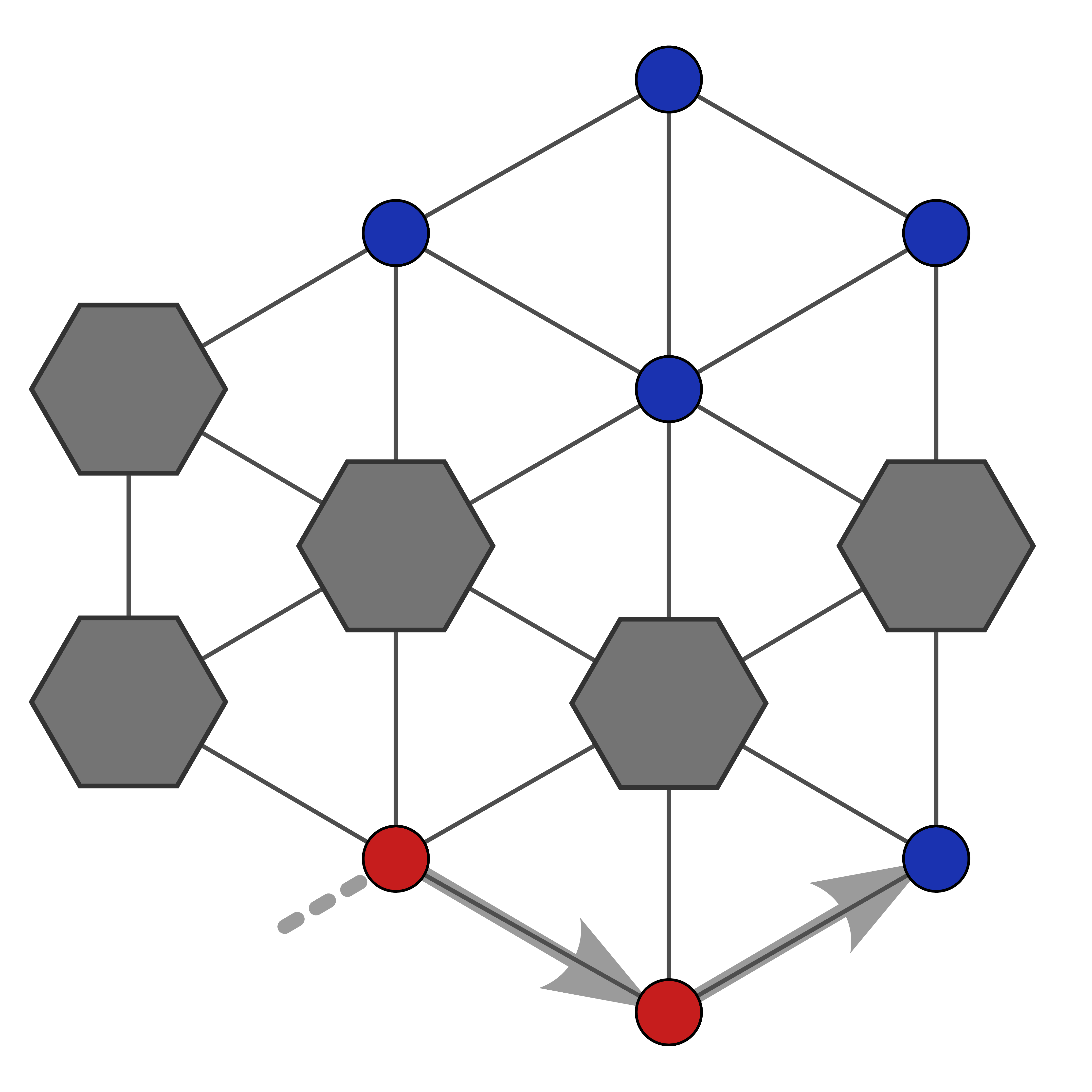}
                \subcaption{step $t^+$}
                \label{fig:invariantGen2}
            \end{subfigure}
            \caption{Local configuration in the proof of Claim 4. Note that while $\lp_{i-1}$ is depicted as a link (circular, red), it may be the node $\start{}$ instead.}
            \label{fig:coatingLayerGraphasd}
        \end{minipage}%
        \hfill
        \begin{minipage}[t][][b]{.315\textwidth}
            \centering
            \captionsetup[subfigure]{labelformat=empty}
            \begin{subfigure}[c]{\linewidth}
                \includegraphics[width=\linewidth]{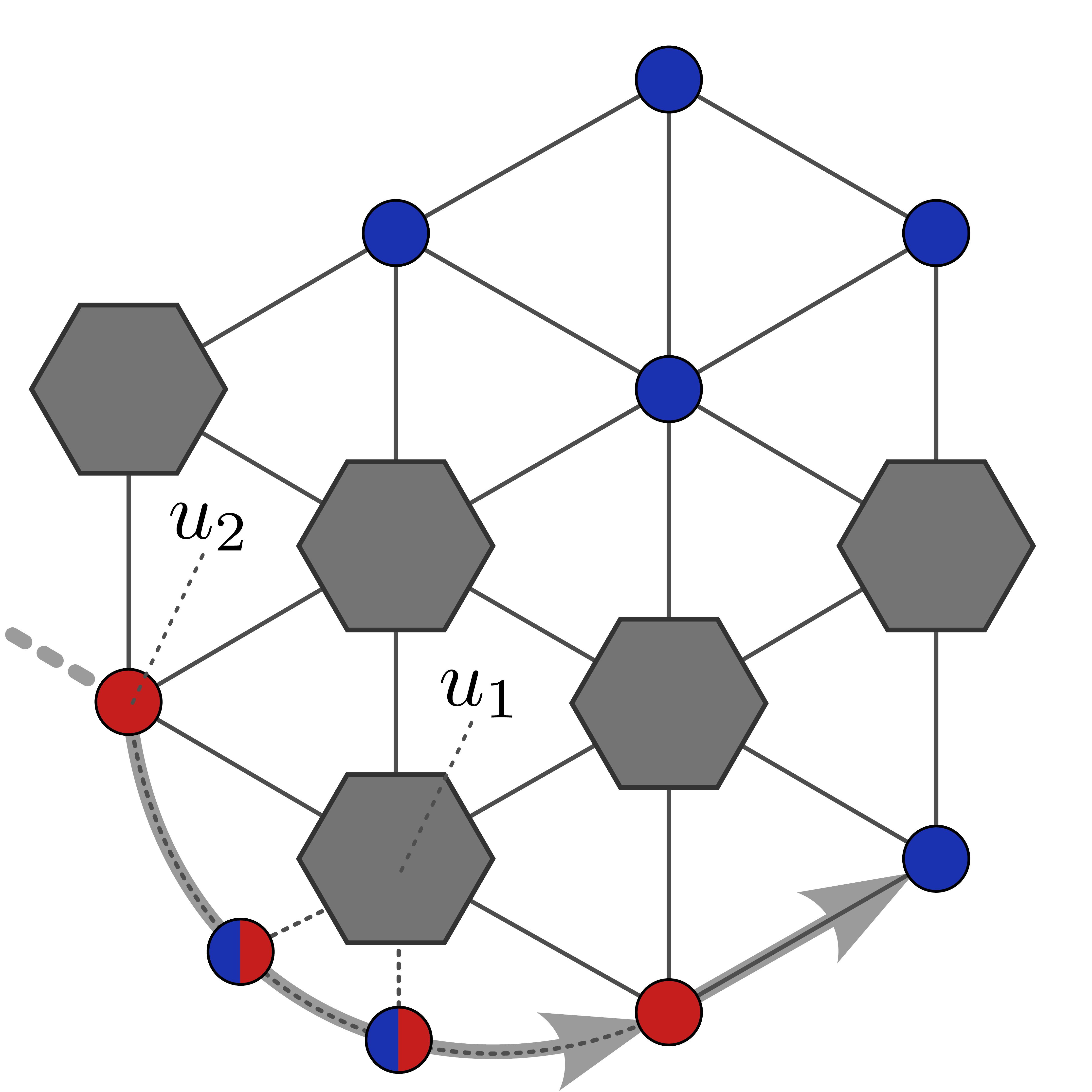}
                \subcaption{}
                \label{fig:invariantGen3}
            \end{subfigure}
            \caption{Local configuration after a tile is placed at $u_1$ with $\noCheck = true$.}
            \label{fig:boundaryasd}
        \end{minipage}
    \end{figure}
    \begin{claim}  
        If $P1$--$P5$ hold in step $t$, and a tile is placed at some $\lp_i \in \links \cap \gen$ with $\lp_{i-1} \in \links \cup \{\start{}\}$, then $P1$--$P5$ hold in step $t^{+}$.
    \end{claim}
    \begin{claimproof}
        Since a tile is placed at a generator, by \cref{lem:enterCoat} the agent enters phase \coat{} at $\lp_{i+2}$ such that $\lp_{i+1}, \lp_i \in \links \cap \gen$.
        It follows that the neighborhood of $\lp_i$ and $\lp_{i+1}$ is the same as depicted in \cref{fig:noCheck1}.
        However, since $\lp_{i-1} \in \links \cup \{\start{}\}$, the neighborhood of $\anchor(\lp_i)$ differs.
        If $\anchor(\lp_i)$ has only one tiled neighbor, then the proof reduces to the proof of Claim 3 except that $\noCheck$ is not set to $true$ and $P2$ holds directly in step $t^+$.
        Hence, we must only consider the case $|\bo(\anchor(\lp_i))| = 2$ which is precisely depicted in \cref{fig:invariantGen1} apart from rotation.

        First, assume that $\lp_{i-1} = \start{}$.
        By \cref{lem:enterCoat}, $\lp_{i+2}$ is the first node $v \in \lp$ with $v\notin \links \cup \{\start{}\}$.
        Hence, $\lp_i$ and $\lp_{i+1}$ are the only existing links.
        Let $v$ be the node generated by $\lp_i$ and $w$ the node in $\be(\lp_i)$ that is not $\lp_{i+1}, \lp_i-1$ or $v$.
        As can be seen in \cref{fig:invariantGen1}, $\lp_i$ consumes $\lp_{i+1}$, it generates $v$, and by \cref{lem:generatorSharedTiled} cannot generate $w$.
        Then in step $t^+$ (see \ref{fig:invariantGen2}), $\lp^+$ is given by $\lp^+ = (\start{},v,w,...)$ with $\links^+ = \{v\}$ and $\suc(w) = w \notin \links^+$ such that $P1$--$P5$ hold.

        Second, assume that $\lp_{i-1} \neq \start{}$, i.e., by the lemma's assumption $\lp_{i-1} \in \links$.
        Note that $\be(\lp_{i-1})$ and $\bo(\lp_{i-1})$ both contain a connected component of size at least two.
        Since $\lp_{i-1}$ is empty, it holds that $|B(\lp_{i-1})| \leq 6$, and since $\lp_{i-1}$ is a link, $\be(\lp_{i-1})$ and $\bo(\lp_{i-1})$ must each contain another component of size one.
        Especially, $\bo(\lp_{i-1})$ contains no connected component of size larger than two.

        In the proof of Claim 3, we showed that apart from rotation there is only one local configuration in which a tile is placed with $\noCheck = true$.
        For ease of reference, the configuration after the tile is placed is depicted in \cref{fig:invariantGen3} with new labeling on the nodes.
        Let $u_1$ be the node at which a tile is placed with $\noCheck = true$ in that case, and $u_2$ the node from which the agent moves to $u_1$ before placing the tile.
        As can be seen in \cref{fig:invariantGen3}, $\bo(u_2)$ contains a connected component of size at least three.
        By contraposition, $\anchor(\lp_i)$ must have been placed with $\noCheck = false$ since we showed that $\bo(\lp_{i-1})$ cannot contain a connected component of size larger than two.
        In that case, the proof again reduces to the proof of Claim 3 as described above.
    \end{claimproof}
    Our claims cover all cases in which the agent places a tile according to \cref{alg:algorithm} without terminating afterwards.
    Hence, the case distinction is complete and the lemma follows.
\end{proof}

\end{document}